%% file: Draft_v2_unblinded.tex
\documentclass[12pt]{article} 
\input{Notation}

\usepackage{enumerate}

\usepackage[nodisplayskipstretch]{setspace}

\newcommand*\samethanks[1][\value{footnote}]{\footnotemark[#1]}

\title{Simultaneous sparse estimation of canonical vectors in
the $p \gg  N$ setting}
\author{Irina Gaynanova\thanks{Department of Statistical Science, Cornell University,
Ithaca, NY; e-mail: \textsf{ig93@cornell.edu}},
 James G. Booth\thanks{Department of Biological Statistics and Computational Biology, Cornell University,
Ithaca, NY} and Martin T. Wells\samethanks[1]}
\date{}

\newcommand\blfootnote[1]{%
  \begingroup
  \renewcommand\thefootnote{}\footnote{#1}%
  \addtocounter{footnote}{-1}%
  \endgroup
}

\usepackage{algorithm}
\usepackage{algpseudocode}

\begin{document}
\maketitle
\begin{singlespace}
\begin{abstract}
This article considers the problem of sparse estimation of canonical vectors in linear discriminant analysis when $p\gg N$. Several methods have been proposed in the literature that estimate one canonical vector in the two-group case. However, $G-1$ canonical vectors can be considered if the number of groups is $G$. In the multi-group context, it is common to estimate canonical vectors in a sequential fashion. Moreover, separate prior estimation of the covariance structure is often required. We propose a novel methodology for direct estimation of canonical vectors. In contrast to existing techniques, the proposed method estimates all canonical vectors at once, performs variable selection across all the vectors and comes with theoretical guarantees on the variable selection and classification consistency. First, we highlight the fact that in the $N>p$ setting the canonical vectors can be expressed in a closed form up to an orthogonal transformation. Secondly, we propose an extension of this form to the $p\gg N$ setting and achieve feature selection by using a group penalty. The resulting optimization problem is convex and can be solved using a block-coordinate descent algorithm. The practical performance of the method is evaluated through simulation studies as well as real data applications. 
\end{abstract}
\end{singlespace}
\blfootnote{
This research was partially supported by NSF-DMS 1208488, NSF-DMS 0808864 and ASAF grant FA9550-13-1-0137. We are grateful to Jacob Bien for a valuable discussion of block-coordinate descent algorithms.}


\section{Introduction}
Recent technological advances have generated high-dimensional data sets across a wide variety of application areas such as finance, atmospheric science, astronomy, biology and medicine. Not only do these data sets provide computational challenges, but they also breed new statistical challenges as the traditional methods no longer sufficient. Linear Discriminant Analysis (LDA) is a popular classification and data visualization tool that is used in the $N\gg p$ setting. LDA seeks the linear combinations of features that maximize Between Group Variability with respect to Within Group Variability \citep[Chapter~11]{Mardia:1979vm}. These linear combinations are called \textit{canonical vectors} and they provide a low-dimensional representation of the data by reducing the original feature space dimension $p$ to $G-1$, where $G$ is the total number of groups. 

The classical use of LDA when $p\gg N$ fails to provide useful results because of the singularity of covariance matrix and over-selection of relevant features \citep{Dudoit:2002ev, Bickel:2004wa}. As a result, the extension of LDA to high-dimensional settings has recently received a lot of literature attention. A number of these proposals result in non-sparse classifiers. \citet{Friedman:1989tm}, \citet{Krzanowski:1995ex} and \citet{Xu:2009fl} regularize the within-class covariance matrix in order to obtain a positive definite estimate. Other approaches that lead to sparse discriminant vectors have also been considered. \citet{Tibshirani:2002ht} propose the shrunken centroids methodology by adapting the naive Bayes classifier and soft-thresholding the mean vectors. \citet{Guo:2007te} combine the shrunken centroids approach with a ridge-type penalty on the within-class covariance matrix.  \citet{Witten:2011kc} apply an $\ell_1$ penalty to the Fisher's discriminant problem in order to obtain sparse discriminant vectors. \citet{Clemmensen:2011kr} use an optimal scoring approach which essentially reduces the sparse discriminant vector construction to a penalized regression problem.

In the two-group setting, \citet{Cai:2011dm} and \citet{Mai:2012bf} propose direct estimation of the canonical vector thus avoiding separate estimation of the covariance matrix. Simulations and real data applications show that the direct estimation approach results in reduced misclassification rates in comparison to alternative methods. The corresponding optimization problems can be solved efficiently for large data sets and have desirable theoretical properties. Unfortunately, the extension of two-group methods to the multi-group case is nebulous \citep[p. 658]{Hastie:2009fg}. Popular approaches include ``one-versus-all" and ``one-versus-one" methods, where the final classification assignment is usually based on the ``majority vote''. As such, computation of more vectors is required ($G$ and $G(G-1)/2$ versus $G-1$).

\citet{Witten:2011kc} and \citet{Clemmensen:2011kr} propose estimating canonical vectors in a sequential fashion in the multi-group setting: starting with the first canonical vector $v_1$, with subsequent $v_i$ found subject to orthogonality constraints. This approach is undesirable  from a computational viewpoint, as well as from an estimation perspective. Each subsequent canonical vector $v_i$ relies on all the previous estimates $v_k$ for $k<i$, hence there is a propagation of the estimation error. In addition, the corresponding optimization problems are nonconvex, hence the convergence of the optimization algorithms to the global solution is not assured. This computational burden poses additional theoretical challenges in the analysis and, as a result, these methods do not come with theoretical guarantees.

The objective of this article is to bridge the computational and theoretical gap in the literature between the two-group and multi-group methods. Inspired by the superior performance of the direct estimation methods by \citet{Cai:2011dm} and \citet{Mai:2012bf} in the two-group case, our goal is to introduce and develop a novel methodology that has the same guaranteed performance in the multi-group setting. Our proposal is based on the observation that canonical vectors can be expressed in a closed form up to an orthogonal transformation. Moreover, this transformation affects neither the classification rule nor the sparsity pattern.  
The definitive contribution of this work is the development of novel methodology that provides:
\begin{enumerate}[(a)]
\item simultaneous estimation of all $G-1$ canonical vectors without prior estimation of the covariance structure;
\item simultaneous variable selection from all canonical vectors;
\item theoretical guarantees on variable selection and classification consistency in multi-group settings.
\end{enumerate}
To our knowledge, this is the first method for multi-group sparse discriminant analysis that achieves all of these goals. 

In addition, while the motivation for our approach is quite different from~\citet{Mai:2012bf}, we show that the two methods are equivalent in the two-group setting. We use this connection to extend the theoretical results of \citet{Mai:2012bf} to the multi-group setting and show that the proposed method can consistently identify the true support of canonical vectors.

The proposed optimization problem is convex and therefore can be solved efficiently for large data sets. Our algorithm doesn't require additional regularization of the sample covariance matrix or generation of an initial starting point, as the form of the optimization problem suggests a natural choice for the grid of tuning parameters. This is not the case for the methods of  \citet{Cai:2011dm}, \citet{Witten:2011kc} or \citet{Clemmensen:2011kr}, where the appropriate grid must be carefully chosen by the user.  While these advantages are mostly computational, they simplify the implementation of the method leading to more consistent results across users. More implementation details are provided in Section~\ref{s:implementation}.

The rest of this paper is organized as follows. Section \ref{sec:WD} discusses the canonical vectors estimation problem and provides new insights into the form of the canonical vectors in the standard $N\gg p$ setting. We use these insights to propose a direct estimation procedure and describe computational aspects of the algorithm. Section~\ref{s:theory} develops bounds on the estimation error and proves that the proposed method identifies the true support of canonical vectors with a high probability. Section \ref{sec:sim} provides simulation results while Section \ref{sec:data} describes applications to real datasets. We conclude with discussion of future research in Section~\ref{sec:disc}.

\section{Methodology}\label{sec:WD}
\subsection{Notation}\label{sec:note}
For a vector $b \in \mathbb{R}^p$ we define $\|b\|_{\infty}=\max_{i=1,...,p}|b_i|$, $\|b\|_1=\sum_{i=1}^p |b_i|$ and $\|b\|_2=\sqrt{\sum_{i=1}^pb^2_i}$.
For a matrix $M\in \mathbb{R}^{n\times p}$ we define $m_i$ to be its $i$th row and $M_j$ to be its $j$th column.  We also define $\|M\|_{\infty}=\max_{i=1,...,n}\|m_i\|_1$, $\|M\|_{\infty,2}=\max_{i=1,..,n}\|m_i\|_2$, $\|M\|_{1}=\sum_{i=1}^{n}\sum_{j=1}^{p}|m_{ij}|$, $\|M\|_{F}=\sqrt{\sum_{i=1}^n\sum_{j=1}^{p}m_{ij}^2}$ and $\|M\|_{*}=\sum_{i=1}^{\min(n,p)}\sigma_{i}(M)$, where $\sigma_{i}(M)$ is the $i$th singular value of $M$. We define $\mathbb{O}^p$ to be the space of $p\times p$ orthogonal matrices $R$ such that $RR^t=R^tR=I$.
\subsection{Estimation problem}
We assume that $X_i\in \mathbb{R}^{p}$, $i=1,...,N$, are independent and come from $G$ groups with different means and the same covariance matrix, i.e. $\E(X_i|Y_i=g)=\mu_{g}$ and $\cov(X_i|Y_i=g)=\varSigma_W$,  where $Y_i\in \{1,...,G\}$.  The between-group population covariance matrix $\varSigma_B$ is defined as 
\begin{equation*}
\varSigma_B=\sum_{g=1}^G\pi_g(\mu_g-\mu)(\mu_g-\mu)^t,
\end{equation*}
 where $\pi_g=P(Y_i=g)$ are group-specific probabilities and $\mu=\sum_{i=1}^G\pi_g\mu_g$ is the overall population mean. The \textit{population canonical vectors} $\varPsi$ are defined as eigenvectors corresponding to non-zero eigenvalues of $\varSigma_W^{-1}\varSigma_B$. Although the eigenvectors are unique only up to normalization \citep{Golub:2012wp}, we take advantage of  the uniqueness of the eigenspace in defining a scale-invariant classification rule. For a new observation value of $X$, $x\in\R^{p}$, the population classification rule $h_{\varPsi}(x)$ is defined as
\begin{equation}\label{eq:prule}
h_{\varPsi}(x)=\arg\min_{1\le g\le G}(x-\mu_{g})^{t}\varPsi(\varPsi^{t}\varSigma_{W}\varPsi)^{-1}\varPsi^{t}(x-\mu_{g})-2\log\pi_{g}.
\end{equation}
The classification rule is based on the closest Mahalanobis distance in the projected space defined by $\varPsi$, after adjustment for potential discrepancy in the prior group probabilities $\pi_{g}$. Through the addition of $2\log \pi_{g}$ term, the resulting classification rule mimics the optimal classification rule under the assumption that the data comes from the multivariate normal group-conditional distribution \citep[Chapter~3.9.3]{McLachlan:1992dl}.  Our goal is to identify the eigenspace spanned by $\varPsi$ based on the sample observations $X_i \in \mathbb{R}^{p}$ and sample labels $Y_i\in \{1,...,G\}$. 

Consider the within-group sample covariance matrix $W=\frac1{N-G}\sum_{g=1}^{G}(n_g-1)S_g$ and the between-group sample covariance matrix $B=\frac{1}{N}\sum_{g=1}^Gn_g(\bar X_g-\bar X)(\bar X_g-\bar X)^t$, where $n_g$ is the number of observations in group $g$, $S_g$ is the sample covariance matrix for group $g$, $\bar X_g$ is the sample mean for group $g$ and $\bar X$ is the overall sample mean. Recall that $W$ is nonsingular when $N\gg p$ and therefore we can define the \textit{sample canonical vectors} $V$ as $G-1$ eigenvectors corresponding to non-zero eigenvalues of $W^{-1}B$ \citep[Chapter~11.5]{Mardia:1979vm}. Similarly to~\eqref{eq:prule}, the sample classification rule $\hat h_{V}(x)$ is defined as
\begin{equation}\label{eq:srule}
\hat h_{V}(x)=\arg\min_{1\le g\le G}(x-\bar X_{g})^{t}V(V^{t}WV)^{-1}V^{t}(x-\bar X_{g})-2\log\frac{n_{g}}{N}.
\end{equation}
In what follows we show that in the $N\gg p$ setting, canonical vectors can be expressed in a closed form up to an orthogonal transformation (Proposition \ref{cl:V}). Since this transformation has no effect on the classification rule (Proposition \ref{cl:clas}), it allows us to estimate $\varPsi$ directly.
\begin{Claim}\label{cl:BD}
 The following decompositions hold: $\varSigma_B=\Delta\Delta^t$ and $B=DD^t$, where for $r=1,...,G-1$ the $r$th column of $\Delta$ has the form 
\begin{equation}\label{eq:varDelta}
 \varDelta_{r}=\frac{\sqrt{\pi_{r+1}}\left(\sum_{i=1}^{r}\pi_i(\mu_i-\mu_{r+1})\right)}{\sqrt{\sum_{i=1}^{r}\pi_i\sum_{i=1}^{r+1}\pi_i}}
\end{equation}
and the $r$th column of $D$ has the form 
\begin{equation}\label{eq:D}
 D_{r}=\frac{\sqrt{n_{r+1}}\left(\sum_{i=1}^{r}n_i(\bar X_i-\bar X_{r+1})\right)}{\sqrt{N}\sqrt{\sum_{i=1}^{r}n_i\sum_{i=1}^{r+1}n_i}}.
\end{equation}
\end{Claim}
\noindent
The low-rank decomposition of matrices $\varSigma_{B}$ and $B$ is not unique. Our choice of $\Delta$ and $D$ in Proposition~\ref{cl:BD} is motivated by the fact that these matrices can be expressed in a closed form (unlike the eigenvectors of $\varSigma_{B}$ and $B$) and have intuitive interpretation in terms of the differences between the group means. Specifically, the columns of $\Delta$ and $D$ define orthogonal contrasts between the means of $G$ groups. In the case $G=2$, $\varDelta=\sqrt{\pi_1\pi_2}(\mu_2-\mu_1)$ and $D=\frac{\sqrt{n_1n_2}}{N}(\bar X_2 - \bar X_1)$.

\begin{Claim}\label{cl:V}
 Define $\Delta$ and $D$ as in \eqref{eq:varDelta} and \eqref{eq:D}. There exists a matrix $\Rho \in \mathbb{O}^{G-1}$ such that $\varPsi=\varSigma_W^{-1}\Delta \Rho.$ Moreover, if $W$ is nonsingular, there exists a matrix $R\in \mathbb{O}^{G-1}$ such that $V=W^{-1}DR.$
\end{Claim}

\begin{Claim}\label{cl:clas}
The population classification rule based on $\varPsi$ is the same as the population classification rule based on $\tilde \varPsi=\varSigma_W^{-1}\Delta$: $h_{\varPsi}(x)=h_{\tilde\varPsi}(x)$ for all $x\in \R^{p}$. If $W$ is nonsingular, then $\hat h_{V}(x)=\hat h_{W^{-1}D}(x)$ for all $x\in \R^{p}$.\end{Claim}

\subsection{Proposed Estimation Criterion}\label{sec:pn}
From Propositions \ref{cl:V} and \ref{cl:clas} it follows that for classification it is sufficient to estimate 
\begin{equation}\label{eq:tildepsi}
\tilde \varPsi=\varSigma_W^{-1}\varDelta
\end{equation}
 rather than $\varPsi$. 
To illustrate the motivation behind the proposed optimization problem, we first discuss our choice of the loss function and then our choice of the penalty.

Our first goal is to choose a suitable loss function that will capture the deviations of the estimator from the target $\tilde \varPsi=\varSigma_W^{-1}\varDelta$. 
We note that $\tilde \varPsi$ in \eqref{eq:tildepsi} can be defined as
\begin{equation}\label{eq:loss2}
\tilde \varPsi=\arg\min_{V\in\R^{p\times(G-1)}}\frac12\|\varSigma_W^{1/2} V-\varSigma_{W}^{-1/2}\varDelta\|^2_F=\arg\min_{V\in\R^{p\times(G-1)}}\frac12\Tr\left(V^t\varSigma_WV-2\varDelta^tV\right).
\end{equation} 
In the two-group case, $V$ is a vector and the objective function in \eqref{eq:loss2} reduces to 
$$\frac12(\varSigma_W^{1/2} V-\varSigma_{W}^{-1/2}\varDelta)^t(\varSigma_W^{1/2} V-\varSigma_{W}^{-1/2}\varDelta)=\frac12(V-\varSigma_W^{-1}\varDelta)^t\varSigma_W(V-\varSigma_W^{-1}\varDelta)=\frac12(V-\tilde \varPsi)^t\varSigma_W(V-\tilde \varPsi).$$
 This objective function is the same as the quadratic loss function considered by \citet{Rukhin:1992it}, who observed that it is invariant with respect to linear transformation of the data. Hence, we can define an estimator $\tilde V$ by substituting $\varSigma_W$ and $\varDelta$ with $W$ and $D$:
\begin{equation}\label{eq:sampleLoss}
  \tilde V=\arg\min_{V \in \mathbb {R}^{p \times (G-1)}}\frac12\Tr\left(V^tW V-2D^tV\right).
 \end{equation}
 
 Our second goal is to perform a variable selection, which in the discriminant analysis framework corresponds to estimating some entries of $\varPsi$ exactly as zero. A common penalty that is used in this context is an $\ell_1$ penalty $\|V\|_1=\sum_{i=1}^p\sum_{j=1}^{G-1}|v_{ij}|$. Although this penalty leads to sparse estimates, there is no guarantee that the features are simultaneously eliminated from all canonical vectors. In other words, although each canonical vector is sparse individually, the total number of features used may be very large. Moreover, the sparsity of the individual matrix elements is not preserved under the orthogonal rotation. Hence, the sparsity of elements of $\tilde \varPsi$ doesn't imply the sparsity of elements of $\varPsi$.

To overcome the overfitting and nonorthogonality issues, we consider the row-wise penalty $\sum_{i=1}^p\|v_i\|_2$. This penalty results in an estimate that is invariant to orthogonal transformation and eliminates features from all canonical vectors at once inducing row sparsity on the matrix $V$. Alternative penalties that achieve this goal include group SCAD and group MCP, we refer the reader to \citet{Huang:2012wg} for an overview. Our choice of $\sum_{i=1}^p\|v_i\|_2$ is motivated by the fact that it preserves convexity of the underlying optimization problem. Combining \eqref{eq:sampleLoss} with this penalty suggests an estimator $\hat V(\lambda)$, defined as
 \begin{equation}\label{eq:test1}
  \hat V(\lambda)=\arg \min_{V \in \mathbb{R}^{p \times (G-1)}}\frac12\Tr\left(V^tW V-2D^{t}V\right)+\lambda\sum_{i=1}^p\|v_i\|_2.
   \end{equation}
 When $W$ is nonsingular and $\lambda=0$, $\hat V(\lambda)=W^{-1}D$, which according to Proposition \ref{cl:V} is the matrix of sample canonical vectors up to an orthogonal rotation. Unfortunately, the objective function in~\eqref{eq:test1} can be unbounded when $W$ is singular since $\Tr\left(V^tW V-2D^{t}V\right)$ can be made arbitrarily small. Hence, an additional regularization of \eqref{eq:test1} is required.
 
 A simple solution is to use $\tilde W=W+\rho I$ instead of $W$ in \eqref{eq:test1}. This type of regularization is quite common in the LDA context and is used by \citet{Friedman:1989tm}, \citet{Guo:2007te} and \citet{Cai:2011dm}. In our case it leads to
  \begin{equation*}
  \begin{split}
  \hat V(\lambda,\rho)&=\arg \min_{V \in \mathbb{R}^{p \times G-1}}\Tr\left(\frac12V^t\tilde WV-D^tV\right)+\lambda\sum_{i=1}^p\|v_i\|_2\\
  &=\arg \min_{V \in \mathbb{R}^{p \times G-1}}\frac12\Tr(V^t WV)+\frac{\rho}2\|V-D\|^2_F+\lambda\sum_{i=1}^p\|v_i\|_2.
  \end{split}
 \end{equation*}
The second component of the objective function encourages $\hat V(\lambda,\rho)$ to be close to $D$, especially when $\rho$ is large. In contrast, $\hat V(\lambda,\rho)$ should be close to $W^{-1}D$ according to Proposition \ref{cl:V}. This discrepancy suggests that strong regularization of $W$ may have a negative affect on classification performance.

We consider an alternative approach to lower bound the objective function in~\eqref{eq:test1} by substituting $W$ in \eqref{eq:test1} with $T=W+DD^{t}=W+B$.  Such a substitution is possible because the matrices $W^{-1}B$  and $(W+B)^{-1}B$ have the same eigenvectors corresponding to non-zero eigenvalues (see Proposition~\ref{cl:eigen}). In addition, the substitution preserves the functional form of the objective in~\eqref{eq:test1} and helps us establish the connection with the previous two-group LDA methods (Section~\ref{s:connection}). The resulting estimator has the form
\begin{equation}\label{eq:Frob2}
  \hat V (\lambda)=\arg\min_{V \in \mathbb {R}^{p \times (G-1)}}\frac12\Tr\left(V^tW V\right)+\frac12\|D^tV-I\|^2_F+\lambda\sum_{i=1}^{p}\|v_{i}\|_{2},
 \end{equation}
where the objective function is convex and bounded below by zero. The three components of the objective function in \eqref{eq:Frob2} attempt to minimize the within-group variability, control the level of the between-group variability and provide regularization by inducing sparsity respectively.

\begin{Claim}\label{cl:eigen}
 Let $\varPsi$ be the matrix of eigenvectors corresponding to non-zero eigenvalues of $\varSigma^{-1}_W\varSigma_B$ and $\varUpsilon_{\rho}$ be the matrix of eigenvectors corresponding to non-zero eigenvalues of $\left(\varSigma_W+\rho\varSigma_B\right)^{-1}\rho\varSigma_B$ for some positive $\rho$. Then there exists a diagonal matrix $K_{\rho}$ such that $\varPsi=\varUpsilon_{\rho} K_{\rho}.$
\end{Claim}
  
\subsection{Connection with other sparse discriminant analysis methods when $G=2$}\label{s:connection}
The motivation for our method is based on the eigenstructure of the discriminant analysis problem in the multi-group setting, however it has a direct connection with the two-group methods previously proposed in the literature. When $G=2$, $V$ is a vector in $\mathbb{R}^p$ and \eqref{eq:Frob2} takes the form
\begin{equation*}
  \hat V(\lambda)=\arg \min_{V \in \mathbb{R}^p}\frac12V^tWV+\frac12(D^tV-1)^2+\lambda\|V\|_1.
\end{equation*}

\begin{Claim}\label{cl:Mai}
Consider $\hat V_{DSDA}(\lambda)$ \citep{Mai:2012bf}, defined as
\begin{equation*}
\hat V_{DSDA}(\lambda)=\arg\min_{\beta_0\in \mathbb{R},V \in \mathbb{R}^p} \frac1{2N}\sum_{i=1}^N(y_i-\beta_0-X_i^tV)^2+\lambda\|V\|_1,
\end{equation*}
where $y_i=-\frac{N}{n_1}$ if  the $i$th subject is in group 1 and $y_i=\frac{N}{n_2}$ otherwise.
Then
 $$\hat V(\lambda)=\frac{N}{\sqrt{n_1n_2}}\hat V_{DSDA}\left(\frac{N}{\sqrt{n_1n_2}}\lambda\right).$$
\end{Claim}

Furthermore, \citet{Mai:2012ea} show an equivalence between the three methods for sparse discriminant analysis in the two-group setting: \citet{Wu:2009vk}, \citet{Clemmensen:2011kr} and \citet{Mai:2012bf}. It follows that our method belongs to the same class, however it can be applied to any number of groups. Thus, it can be viewed as a multi-group generalization of this class of methods.

The optimization problem in \eqref{eq:Frob2} corresponds to the choice of $\rho=1$ in Proposition~\ref{cl:eigen}. In general, any $\rho>0$ leads to 
\begin{equation}\label{eq:T}
  \hat V(\lambda,\rho)=\arg \min_{V \in \mathbb{R}^{p \times G-1}}\frac12\Tr(V^tWV)+\frac{\rho}2\|D^tV-I\|^2_F+\lambda\sum_{i=1}^p\|v_i\|_2.
 \end{equation}
When $\rho \to \infty$, \eqref{eq:T} is equivalent to 
\begin{equation}\label{eq:Tinf}
  \hat V(\lambda,\rho=\infty)=\arg \min_{D^tV=I}\frac12\Tr(V^tWV)+\lambda\sum_{i=1}^p\|v_i\|_2,
 \end{equation}
hence the optimization problem \eqref{eq:T} can be considered a convex relaxation to \eqref{eq:Tinf} for large values of $\rho$. When the number of groups is two, the optimization problem \eqref{eq:T} is equivalent to the proposal of \citet{Fan:2012ud}, who also observe the connection between ~\eqref{eq:T} and~\eqref{eq:Tinf}. They perform a simulation study to assess the effect of the tuning parameter $\rho$ and note that its value doesn't significantly affect the classification results as long as the best $\lambda$ is chosen for each $\rho$. They keep the value of $\rho$ at a fixed level $\rho=10$. 

\subsection{Optimization Algorithm}\label{s:algorithm}
The optimization problem in \eqref{eq:Frob2} is convex with respect to $V$ and therefore can be solved efficiently using a block-coordinate descent algorithm. Alternative algorithms include proximal gradient methods and interior-point methods, we refer the reader to \citet{Bach:2011ty} for the overview of convex optimization with sparsity-inducing norms. We chose to use the block-coordinate descent algorithm as it takes advantage of warm starts when solving for a range of tuning parameters and is one of the fastest algorithms for smooth losses with separable regularizers  \citep{Bach:2011ty, Qin:2013dg}.
 Define 
\begin{equation}\label{eq:Tdef}
T=W+B=W+DD^T.
\end{equation} 
By convexity, the solution to \eqref{eq:Frob2} satisfies the KKT conditions \citep[Chapter~5.5]{Boyd:2004uz}. Differentiating \eqref{eq:Frob2} with respect to the $(G-1) \times 1$ vector $v_j$ formed by the $j$th row of $V$ leads to
\begin{equation}\label{eq:KKT}
 V^tT_j-d_j+\lambda u_j=0,
\end{equation}
where $T_j$ is the $j$th column of matrix $T$ in \eqref{eq:Tdef}, $d_j$ is a $(G-1) \times 1$ vector formed by the $j$th row of matrix $D$ in \eqref{eq:D} and $u_{j}$ is the subgradient of $\|v_{j}\|_{2}$:
\begin{equation*}
 u_j= \begin{cases}
       \frac{v_j}{\|v_j\|_2}, &\text{if $v_j \neq 0$};\\
       \in \{u: \|u\|_2\le 1\}, &\text{if $v_j = 0$}.
      \end{cases}
\end{equation*}
Solving \eqref{eq:KKT} further with respect to $v_j$ leads to
$
 v_j=\left(d_j-\sum_{i\neq j}t_{ij}v_i-\lambda u_j\right)/t_{jj},
$
where $t_{ij}$ are the elements of matrix $T$. This leads to the block-coordinate descent algorithm. 
\begin{algorithm}[t]
Given: $k=1$, $V^{(0)}$
\begin{algorithmic}
\Repeat
\State $\bar V\gets V^{(k-1)}$
	\For{$j=1$ \textbf{to} $p$}
 	 \State $v^{(k)}_{j} \gets \left(1-\frac{\lambda}{\|d_j-\sum_{i\neq j}t_{ij}\bar v_i\|_2}\right)_{+}\left(d_j-\sum_{i\neq j}t_{ij}\bar v_i\right)/t_{jj}$
	\EndFor
\State $k\gets k+1$
\Until{$k=k_{\max}$ or $V^{(k)}$ satisfies stopping criterion.}
\end{algorithmic}
\caption{Block-coordinate descent algorithm.}
\label{algorithm}
\end{algorithm}
\noindent
Note that if $\lambda \ge \max_{1\le i\le p}\|d_i\|_2$, then $\hat V(\lambda)=0$. Moreover, if $T$ is non-singular, by applying the vectorization operator \eqref{eq:Frob2} can be rewritten as
 \begin{equation*}
  \hat V(\lambda)=\arg \min_{V \in \mathbb{R}^{p \times G-1}}\frac12\left\|\vect(D^tT^{-1/2})-(T^{1/2}\otimes I_{G-1})\vect(V^t)\right\|^2_2+\lambda\sum_{i=1}^p\|v_i\|_2.
 \end{equation*}
This formulation corresponds to a group lasso optimization problem \citep{Yuan:2006wj} with the response vector $\vect(D^tT^{-1/2})$ and the design matrix $T^{1/2}\otimes I_{G-1}$. Due to the form of the design matrix, each block subproblem can be solved in a closed form, making the implementation of block-coordinate descent algorithm straightforward.

\section{Theory}\label{s:theory}
In this section we analyze the variable selection and classification performance of the estimator $\hat V(\lambda)$ defined in \eqref{eq:Frob2}. In Section \ref{s:connection} we established an equivalence between our proposal and the proposal of \citet{Mai:2012bf} for the two-group case. We use this connection to extend the variable selection consistency results of \citet{Mai:2012bf} to the multi-group case. In particular, to prove Theorem~\ref{t:main1}, we derive Lemmas~\ref{l:meand}~and~\ref{l:T} that serve as a multi-group version of Lemma~A1 in~\citet{Mai:2012bf}. We also show that the variable selection consistency implies classification consistency.

Denote the support of $\tilde \varPsi$ by $A=\{j:\|\tilde \psi_j\|_2\neq0\}$, where $\tilde \psi_j$ is the $j$th row of $\tilde \varPsi$, and assume that the support is sparse: $s<<p$ where $s=\card(A)$. Denote the support of $\hat V(\lambda)$ by $\hat A=\{j:\|\hat v_j(\lambda)\|_2\neq 0\}$. Let $\varSigma=\varSigma_W+\varSigma_B$, $\varPsi'=\varSigma^{-1}\Delta$ and note that $\{j:\|\psi'_j\|_2\neq0\}=A$, i.e. $\tilde \varPsi$ and $\varPsi'$ have the same support. Furthermore, let
$\kappa=\|\varSigma_{A^cA}\varSigma_{AA}^{-1}\|_{\infty}$, $\phi=\|\varSigma^{-1}_{AA}\|_{\infty}$, $\varPsi_{\min}=\min_{i\in A} \|\psi'_i\|_2$ and $\delta=\|\varDelta\|_{\infty,2}$, where $\varSigma_{AA}$ is the sub-matrix of $\varSigma$ formed by the intersection of the rows and columns in $A$.  In Theorem \ref{t:main1} we establish lower bounds on $P(A=\hat A)$ and $ P\left(\|\hat V(\lambda)_A-\varPsi'_{A}\|_{\infty,2}\le 2\phi \lambda\right)$.
\begin{thm}\label{t:main1}
 Assume $\kappa<1$ and $\left(X_i|Y_i=g\right)\sim N(\mu_g,\varSigma_W)$, $i=1,..,N$. Then 
 \begin{enumerate}
  \item For any $\lambda>0$ and positive $\epsilon\le\frac{\lambda(1-\kappa)}{(\kappa+1)(\phi\delta+1)+2\phi\lambda}$, $\hat V(\lambda)_{A^C}=0$ with a probability of at least $1-t_1$, where
  \begin{equation*}
   t_1=c_{1}ps\exp(-c_2Ns^{-2}\epsilon^2)+2(G-1)p\exp(-c_3N\epsilon^2).
  \end{equation*}
\item For any $\lambda< \frac{\varPsi_{\min}}{\phi}$ and $\epsilon < \frac{\varPsi_{\min}-\lambda\phi}{\phi(1+\phi\delta+\varPsi_{\min})}$ none of the elements of $\hat V(\lambda)_A$ are zero with a probability of at least $1-t_2$, where
  \begin{equation*}
   t_2=c_{1}s^2\exp(-c_2Ns^{-2}\epsilon^2)+2(G-1)s\exp(-c_3N\epsilon^2).
  \end{equation*}
\item For any positive $\epsilon<\frac{\lambda}{1+\phi\delta+2\phi\lambda}$
\begin{equation*}
 P\left(\|\hat V(\lambda)_A-\varPsi'_{A}\|_{\infty,2}\le 2\phi \lambda\right)\ge 1-c_{1}s^2\exp(-c_2Ns^{-2}\epsilon^2)+2(G-1)s\exp(-c_3N\epsilon^2).
\end{equation*}
\end{enumerate}
\end{thm}
While the motivation for the proposed optimization problem doesn't rely on the normality assumption, the normality assumption does simplify the proof. We discuss possible extensions to the non-normal case in the online supplement in Section~S8. We further use Theorem~\ref{t:main1} to establish variable selection consistency of the estimator $\hat V(\lambda)$ defined in \eqref{eq:Frob2}. Specifically, Theorem~\ref{t:main1} implies asymptotic conditions under which $P(A=\hat A)\to 1$, which coincide with asymptotic conditions for the two-group case \citep{Mai:2012bf}:
\begin{itemize}
\item[(C1)] $N\to \infty$, $p \to \infty$, $G=O(1)$ and $\frac{\log(ps) s^2}{N} \to 0$.
\item[(C2)]  $\sqrt{ \frac{\log(ps) s^2}{N}} << \lambda_N << \varPsi_{\min}$.
\end{itemize}
\smallskip
\begin{corollary}\label{cor}
If (C1) and (C2) hold, then
$P(A=\hat A)\to 1$.
\end{corollary}
We also show that under the same asymptotic conditions the sample classification rule based on $\hat V$ coincides with the population classification rule $h_{\varPsi}$ defined in \eqref{eq:prule}. Let $X\in \R^{p}$ be a new observation with a value $x\in \R^{p}$.
\begin{corollary}\label{cor:clas}
If (C1) and (C2) hold, then
$P(\|\hat V(\lambda)_A-\varPsi'_{A}\|_{\infty,2}\le 2\phi \lambda_{N})\to 1$.
Moreover, if $\lambda_{N}\to 0$, then $P\left(\hat h_{\hat V}(x)=h_{\varPsi}(x)\right)\to 1$.
\end{corollary}

\section{Simulation Results} \label{sec:sim}
In this section we evaluate the performance of the estimator $\hat V(\lambda)$ defined in~\eqref{eq:Frob2} against the alternative methods proposed in the literature. We refer to our proposal as MGSDA for Multi-Group Sparse Discriminant Analysis.  The results reported here concern the case in which the sample size for each group is $n=100$ and the number of features is $p=100$ and $p=800$. The test datasets are the same size as the training datasets and are generated independently. Conditional on group $g$, the samples are drawn independently from the multivariate normal distribution $N(\mu_{g},\varSigma_{W})$. The following structures for $\varSigma_{W}$ are considered in all the simulations: 
\begin{enumerate}
 \item \textbf{Identity:} $\varSigma_W=I$.
 \item \textbf{Equicorrelation:} $\varSigma_W=(\sigma_{ij})_{p \times p}$ with $\sigma_{ii}=1$ and $\sigma_{ij}=0.5$ for $i \neq j$.
 \item \textbf{Autoregressive:} $\varSigma_W=(\sigma_{ij})_{p \times p}$ with $\sigma_{ij}=0.8^{|i-j|}$ for $1\le i,j\le p$.
\item \textbf{Bernoulli:}.  $\varSigma_W=\Omega^{-1}$ with $\Omega=(B+\delta I)/(1+\delta)$. Here $B=(b_{ij})_{p \times p}$ with $b_{ii}=1$ for $1\le i\le p$, $b_{ij}=b_{ji}=0.5\times Ber(1,0.2)$ for $1 \le i\le s_0$, $i < j \le p$ and $b_{ij}=b_{ji}=0.5$ for $s_0+1 \le i\le p$, $i < j \le p$. $\delta$ is taken as $\delta=\max(-\lambda_{min}(B),0)+0.05$ to ensure that $\Omega$ is positive definite.
\item \textbf{Data Based:} $\varSigma_W=(1-\alpha)S+\alpha I$, where $\alpha=0.01$ and $S$ is a sample correlation matrix estimated from the most variable $p=800$ features of Ramaswamy dataset \citep{Ramaswamy:2001hc}. The dataset is available from \url{http://www-stat.stanford.edu/~tibs/ElemStatLearn/}.
\end{enumerate} 
Structures 1-3 have been used in simulation studies in LDA literature \citep{Cai:2011dm, Witten:2011kc, Mai:2012bf}, and the Bernoulli structure was considered by \citet{Cai:2011dm}. We view the Data Based structure as an approximation to a covariance structure that is more realistic in practical settings.

\subsection{The two-group case}\label{s:G=2}

This simulation scenario considers the classification between the two groups with $\mu_1=0_p$ and $\mu_2=(1_{s},0_{p-s})$ for covariance structures 1-4. For covariance structure 5  we take $\mu_2=(d_{s},0_{p-s})$ with $d$ ranging from 0.1 to 0.5 since  in this case the Bayes error is almost zero for $\mu_2=(1_{s},0_{p-s})$. The simulations are performed for the values of $s=10$ and $s=30$ for structures 1-4 and $s=10$ for structure 5. 

\citet{Mai:2012bf} perform extensive simulations to compare their proposal with the methods of \citet{Wu:2009vk}, \citet{Witten:2011kc}, \citet{Tibshirani:2003bj} and \citet{Fan:2008tu}. In all the settings, the method of \citet{Mai:2012bf} performs the best in terms of misclassification error. Given Proposition \ref{cl:Mai}, we do not compare MGSDA with any of these methods. On the other hand, \citet{Cai:2011dm} also show that their proposal performs the best when compared to \citet{Shao:2011jw}, \citet{Fan:2008tu} and \citet{Tibshirani:2003bj}. To our knowledge, no comparison was performed between the methods of \citet{Mai:2012bf} and \citet{Cai:2011dm}, therefore  in this section we compare our results to the results of  \citet{Cai:2011dm}. We follow the terminology of  \citet{Cai:2011dm} and refer to their method as Linear Programming Discriminant (LPD). We also evaluate the performance of $\tilde \varPsi=\varSigma^{-1}_{W}\varDelta$. We refer to $\tilde \varPsi$ as the Oracle.

We note that the LPD requires additional regularization of the within-group sample covariance matrix: $\tilde W=W+\rho I $. This regularization is needed to generate a feasible starting point for the optimization algorithm. \citet{Cai:2011dm} suggest taking $\rho\le\sqrt{\log p/N}$. In our simulations $N=200$ and therefore $\rho=0.15$ satisfies this requirement for both $p=100$ and $p=800$. We also try $\rho=2$ to examine how the choice of $\rho$ affects the misclassification rate.


The misclassification error rates as percentages over 100 replications for covariance structures 1-4 are reported in Table \ref{t:simg2}. The corresponding number of selected features and the number of false positive features is reported in Table \ref{t:simg2f}. We define the feature $j$ as a false positive if the corresponding component of estimated canonical vector $\hat V$ is non-zero, $\hat v_{j}\neq0$, but $\mu_{1j}-\mu_{2j}=0$. Note that the population canonical vector $\varPsi$ is truly sparse only in the Identity case, it is only approximately sparse in other scenarios. Comparing MGSDA with the best results of LPD show that the methods have similar error rates when the covariance matrix is Identity or Equicorrelation. MGSDA outperforms LPD for the Autoregressive covariance structure, however LPD performs significantly better for the Bernoulli covariance structure when $p=800$. The methods select comparable numbers of features in all scenarios. 

The mean misclassification rates for the Data Based covariance structure are reported in Figure \ref{fig:databG2}. In this case MGSDA performs significantly better than LPD regardless of the choice of $\rho$. The difference in misclassification rates is especially noticeable when the difference in means $d$ is small. 

The error rates of LPD with $\rho=0.15$ and $\rho=2$ are similar for most of the covariance structures, however they are significantly different for Bernoulli structure when $p=800$ and for the Data Based structure. Table~\ref{t:simg2f} reveals that $\rho$ can also have a significant effect on the number of selected features (the difference is especially noticeable when $s=30$ and $p=800$). This suggests that the choice of $\rho$ can significantly affect the performance of the LPD, with smaller values of $\rho$ likely to result in smaller misclassification error. Unfortunately it remains unclear how to choose the optimal $\rho$ in practical settings.

\begin{table}[!h]
\footnotesize
\centering
\caption{Mean misclassification error rates as percentages over 100 replications, $G=2$, standard deviation is given in brackets.}
\label{t:simg2}
\begin{tabular}{|l|c|c|llll|}
  \hline
Covariance & $s$ & $p$ & MGSDA & LPD, $\rho=0.15$ & LPD, $\rho=2$ & Oracle \\ 
  \hline
Identity & 10 & 100 & 6.65(2.07) & 6.75(2.04) & 6.17(1.94) & 5.58(1.89) \\ 
 & 10 & 800 & 7.32(2.09) & 6.84(1.97) & 6.44(1.73) & 5.75(1.56) \\ 
 & 30 & 100 & 0.9(0.77) & 0.67(0.7) & 0.49(0.53) & 0.4(0.5) \\ 
& 30 & 800 & 0.83(0.69) & 1.09(0.86) & 0.46(0.5) & 0.32(0.39) \\ 
  \hline
  Equicorrelation & 10 & 100 & 3.32(1.25) & 3.38(1.7) & 3.02(1.58) & 1.51(0.89) \\ 
  & 10 & 800 & 3.11(1.25) & 2.98(1.39) & 2.79(1.17) & 1.45(0.81) \\ 
& 30 & 100 & 0.55(0.53) & 0.55(0.67) & 0.52(0.73) & 0.06(0.2) \\ 
 & 30 & 800 & 0.27(0.38) & 0.5(0.61) & 0.56(0.77) & 0(0) \\ 
  \hline
  Autoregressive & 10 & 100 & 19.02(2.91) & 20.83(3.17) & 23.88(2.99) & 16.65(2.48) \\ 
 & 10 & 800 & 22.29(3.26) & 23.59(3.35) & 24.5(3.06) & 16.05(2.59) \\ 
& 30 & 100 & 13.72(2.68) & 15.26(2.91) & 15.85(2.84) & 10.97(2.14) \\ 
& 30 & 800 & 16.57(2.71) & 17.23(3.26) & 16.81(2.64) & 11.13(1.95) \\ 
  \hline
  Bernoulli & 10 & 100 & 6.12(1.69) & 5.88(1.48) & 5.75(1.62) & 4.37(1.35) \\ 
 & 10 & 800 & 37.14(6.04) & 17.03(3.4) & 28.62(3.59) & 4.6(1.49) \\ 
 & 30 & 100 & 0.35(0.42) & 0.3(0.38) & 0.22(0.34) & 0.05(0.15) \\ 
  & 30 & 800 & 8.27(2.81) & 3.29(1.4) & 7.17(2.74) & 0.04(0.14) \\ 
   \hline
\end{tabular}

\caption{Mean number of selected features and false positive features over 100 replications, $G=2$, standard deviation is given in brackets.}
\label{t:simg2f}
\begin{tabular}{|l|c|c|lll | lll |}
  \hline
  &&&\multicolumn{3}{c|}{All features}&\multicolumn{3}{c|}{False positives}\\
Covariance & $s$ & $p$ & MGSDA & LPD, $\rho=0.15$ & LPD, $\rho=2$ & MGSDA & LPD, $\rho=0.15$ & LPD, $\rho=2$ \\ 
  \hline
Identity & 10 & 100 & 20(7) & 19(14) & 19(15) & 10(7) & 9(14) & 9(15) \\ 
 & 10 & 800 & 29(16) & 25(26) & 24(29) & 19(16) & 15(26) & 14(29) \\ 
   & 30 & 100 & 40(7) & 38(11) & 34(13) & 11(7) & 9(11) & 4(13) \\ 
 & 30 & 800 & 51(15) & 125(93) & 56(83) & 22(15) & 95(93) & 26(83) \\ 
  \hline
  Equicorrelation & 10 & 100 & 51(5) & 55(12) & 73(10) & 41(5) & 45(12) & 63(10) \\ 
   & 10 & 800 & 84(13) & 90(54) & 128(48) & 74(13) & 80(54) & 118(48) \\ 
  & 30 & 100 & 77(4) & 78(9) & 93(6) & 49(4) & 49(8) & 63(6) \\ 
& 30 & 800 & 147(13) & 112(54) & 177(40) & 119(13) & 84(54) & 147(40) \\ 
  \hline
  Autoregressive & 10 & 100 & 19(6) & 20(14) & 30(21) & 14(6) & 13(13) & 20(21) \\ 
   & 10 & 800 & 32(15) & 23(21) & 27(31) & 26(15) & 16(20) & 17(31) \\ 
   & 30 & 100 & 26(7) & 31(15) & 46(17) & 12(6) & 14(13) & 17(16) \\ 
   & 30 & 800 & 41(20) & 41(51) & 76(92) & 28(19) & 25(49) & 48(91) \\ 
  \hline
  Bernoulli & 10 & 100 & 24(9) & 18(11) & 20(18) & 14(9) & 8(11) & 10(18) \\ 
  & 10 & 800 & 43(33) & 70(83) & 19(14) & 38(31) & 60(83) & 9(14) \\ 
  & 30 & 100 & 43(8) & 39(14) & 33(11) & 14(8) & 9(14) & 3(11) \\ 
& 30 & 800 & 116(32) & 216(117) & 43(18) & 90(31) & 187(116) & 13(18) \\ 
   \hline
\end{tabular}
\end{table}

\begin{figure}[!t]
\centering
\caption{Mean misclassification error rate in percentage over 25 replications for Data Based covariance structure as a function of difference in means $d$, $G=2$.}
\label{fig:databG2}
\makebox{\includegraphics[scale=0.8]{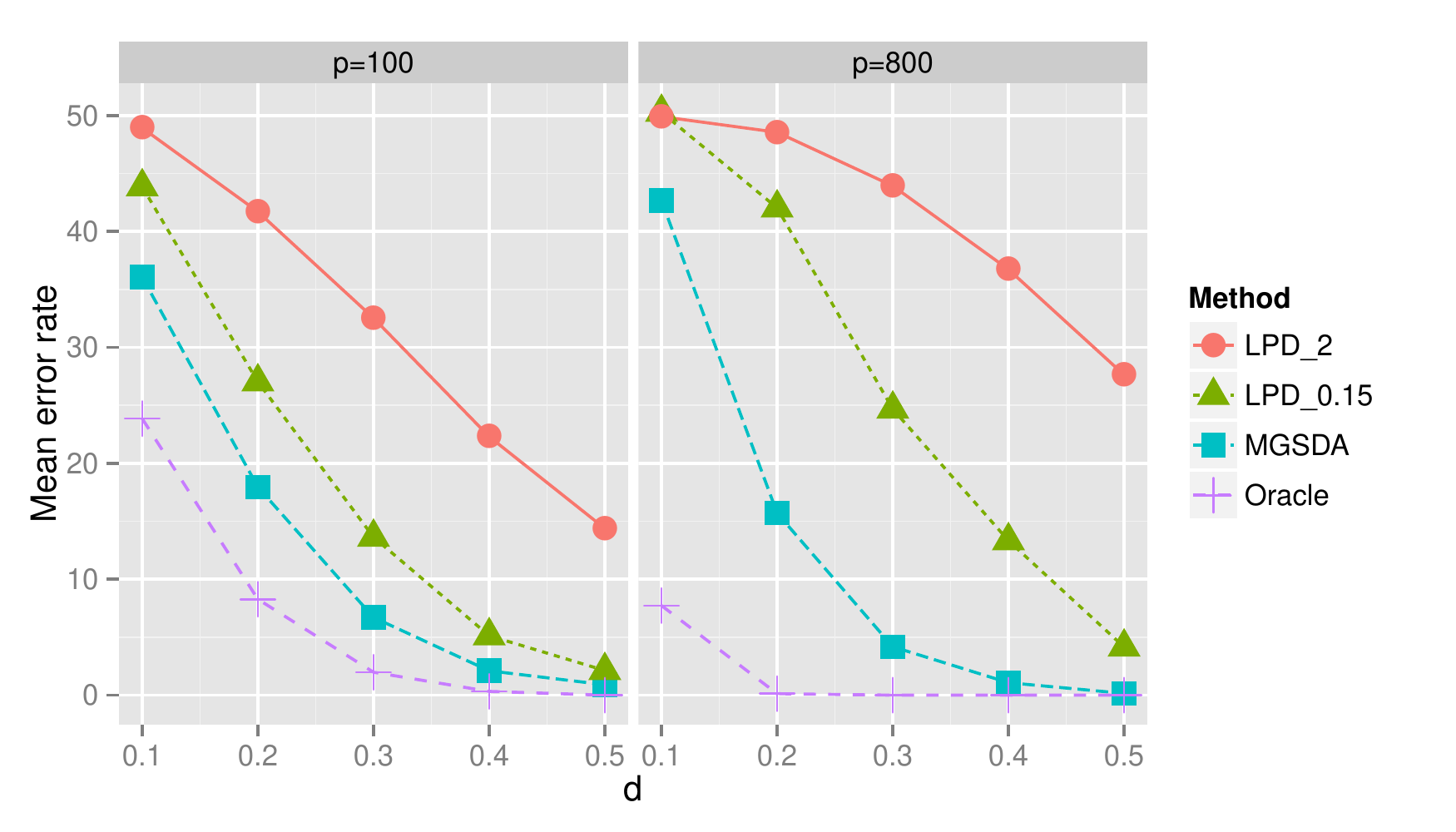}}
\end{figure}

\subsection{The multi-group case}
This simulation scenario considers the classification between the three groups with $\mu_1=0_p$, $\mu_2=(1_{s/2},-1_{s/2},0_{p-s})$ and $\mu_3=(-1_{s/2},1_{s/2},0_{p-s})$. As in the two-group case, we vary the value of $\mu_{2}$ and $\mu_{3}$ for the covariance structure 5: $\mu_2=(d_{s/2},-d_{s/2},0_{p-s})$ and $\mu_3=(-d_{s/2},d_{s/2},0_{p-s})$ with $d$ ranging from 0.1 to 0.5. The simulations are performed for the values of $s=10$ and $s=30$ for structures 1-4 and $s=10$ for structure 5. 

The LPD method of \citet{Cai:2011dm} is developed for the two-group setting. Though it can be generalized to the multi-group case, this generalization is not unique. Among the popular methods are ``one versus one" and ``one versus all" approaches \citep[p. 658]{Hastie:2009fg}. In addition to requiring the computation of a larger number of discriminant vectors ($G(G-1)/2$ and $G$ correspondingly), these approaches can disagree in their classification rules as well as in selected features. Given this ambiguity, we do not compare our method to the LPD in the multi-group case. 

We were able to find only two methods in the literature that specifically consider sparse discriminant analysis in the multi-group case: penalizedLDA by \citet{Witten:2011kc} and sparseLDA by \citet{Clemmensen:2011kr}. Both methods find canonical vectors sequentially and are nonconvex. We compare their performance with MGSDA and also evaluate the performance of $\tilde \varPsi=\varSigma^{-1}_{W}\varDelta$. Again, we refer to $\tilde \varPsi$ as the Oracle.

The mean misclassification error rates as percentages over 100 replications for each combination of parameters are reported in Table \ref{t:simg3}. The number of selected features and the number of false positive features is reported in Table \ref{t:simg3f}. Similar to the two-group case, we define the feature $j$ as a false positive if the corresponding row of estimated canonical vector matrix $\hat V$ is non-zero, $\|\hat v_{j}\|_{2}\neq0$, but $\mu_{1j}=\mu_{2j}=\mu_{3j}$.  Also as before, the population canonical vectors matrix $\varPsi$ is truly row sparse only in the Identity case, it is only approximately row sparse in other scenarios. The results suggest that all three methods are comparable in terms of misclassification rate except for the Autoregressive covariance structure. In this scenario, both MGSDA and sparseLDA outperform the penalizedLDA. In terms of the number of features, MGSDA tends to select fewer than its competitors.   Hence, MGSDA achieves the best tradeoff between the misclassification error and sparsity of the solution.

\begin{table}[!t]
\footnotesize
\centering
\caption{Mean misclassification error rates as percentages over 100 replications, $G=3$, standard deviation is given in brackets.} 
\label{t:simg3}
\begin{tabular}{|l|c|c|llll|}
  \hline
Covariance & $s$ & $p$ & MGSDA & penalizedLDA & sparseLDA & Oracle \\ 
  \hline
Identity & 10 & 100 & 9.11(1.52) & 8.4(1.6) & 9.39(1.68) & 7.83(1.41) \\ 
   & 10 & 800 & 9.22(1.73) & 7.92(1.5) & 9.58(1.87) & 7.67(1.43) \\ 
   & 30 & 100 & 1.06(0.67) & 0.5(0.42) & 0.93(0.63) & 0.42(0.37) \\ 
   & 30 & 800 & 1.43(0.81) & 0.52(0.39) & 1.14(0.72) & 0.43(0.35) \\ 
\hline
  Equicorrelation & 10 & 100 & 2.15(0.97) & 2.34(1.18) & 2.03(0.94) & 1.65(0.86) \\ 
   & 10 & 800 & 2.19(0.89) & 2.12(1.1) & 2.15(0.85) & 1.68(0.84) \\ 
 & 30 & 100 & 0.23(0.38) & 0.3(1.11) & 0.26(0.44) & 0.01(0.05) \\ 
 & 30 & 800 & 0.31(0.43) & 0.04(0.11) & 0.39(0.52) & 0.01(0.06) \\ 
\hline
  Autoregressive & 10 & 100 & 6.83(1.4) & 16.87(2.16) & 6.34(1.32) & 4.87(1) \\ 
 & 10 & 800 & 7.29(1.77) & 16.53(2.44) & 7.47(2.54) & 4.9(1.17) \\ 
 & 30 & 100 & 5.45(1.48) & 16(1.95) & 4.86(1.44) & 3.57(1.05) \\ 
 & 30 & 800 & 5.89(1.53) & 15.46(2.33) & 5.98(2.03) & 3.65(1.05) \\ 
\hline
  Bernoulli & 10 & 100 & 11.15(1.91) & 10.56(1.73) & 11.35(1.88) & 8.51(1.62) \\ 
   & 10 & 800 & 43.13(3.06) & 44.84(4) & 41.72(3.19) & 8.56(1.67) \\ 
 & 30 & 100 & 1.54(0.72) & 1.05(0.56) & 1.37(0.74) & 0.59(0.47) \\ 
  & 30 & 800 & 20.59(3.08) & 22.66(4.11) & 16.42(2.57) & 0.63(0.46) \\ 
   \hline
\end{tabular}

\caption{Mean number of selected features and false positive features over 100 replications, $G=3$, standard deviation is given in brackets.}
\label{t:simg3f}
\begin{tabular}{|l|c|c|lll | lll |}
  \hline
  &&&\multicolumn{3}{c|}{All features}&\multicolumn{3}{c|}{False positives}\\
 Covariance& $s$ & $p$ & MGSDA & penalizedLDA & sparseLDA & MGSDA & penalizedLDA & sparseLDA \\ 
  \hline
Identity & 10 & 100 & 13(7) & 15(15) & 26(11) & 3(7) & 5(15) & 16(11) \\ 
 & 10 & 800 & 11(2) & 15(8) & 24(9) & 1(2) & 5(8) & 14(9) \\ 
& 30 & 100 & 46(18) & 61(19) & 58(11) & 16(18) & 31(19) & 29(11) \\ 
& 30 & 800 & 37(11) & 51(93) & 67(18) & 8(10) & 21(93) & 37(18) \\ 
\hline
  Equicorrelation & 10 & 100 & 14(8) & 10(1) & 27(12) & 4(8) & 0(1) & 17(12) \\ 
 & 10 & 800 & 12(6) & 12(3) & 28(19) & 2(6) & 2(3) & 18(19) \\ 
 & 30 & 100 & 29(11) & 38(12) & 47(8) & 2(10) & 8(12) & 21(6) \\ 
 & 30 & 800 & 29(14) & 30(0) & 50(7) & 4(13) & 0(0) & 25(3) \\ 
\hline
  Autoregressive & 10 & 100 & 21(16) & 12(6) & 27(12) & 13(16) & 2(6) & 20(12) \\ 
& 10 & 800 & 7(3) & 15(11) & 29(19) & 1(2) & 5(11) & 22(19) \\ 
 & 30 & 100 & 28(16) & 54(25) & 39(11) & 13(14) & 24(25) & 22(10) \\ 
& 30 & 800 & 14(5) & 36(40) & 49(24) & 2(4) & 6(40) & 34(23) \\ 
\hline
  Bernoulli & 10 & 100 & 14(10) & 16(15) & 30(13) & 4(10) & 6(15) & 20(13) \\ 
& 10 & 800 & 118(115) & 33(61) & 100(46) & 110(114) & 24(61) & 92(45) \\ 
& 30 & 100 & 51(21) & 59(22) & 61(11) & 22(21) & 29(22) & 31(11) \\ 
 & 30 & 800 & 42(33) & 48(14) & 108(33) & 21(31) & 18(14) & 83(32) \\ 
   \hline
\end{tabular}
\end{table}

The mean misclassification rates for the Data Based covariance structure are reported in Figure~\ref{fig:databG3}. It can be seen that that the penalizedLDA performs significantly worse than both MGSDA and sparseLDA. This is not a surprising result since the Data Based covariance structure is far from diagonal, which is an underlying assumption of penalizedLDA.

\begin{figure}[!t]
\centering
\caption{Mean misclassification error rate in percentage over 25 replications for Data Based covariance structure as a function of difference in means $d$, $G=3$.}
\label{fig:databG3}
\makebox{\includegraphics[scale=0.8]{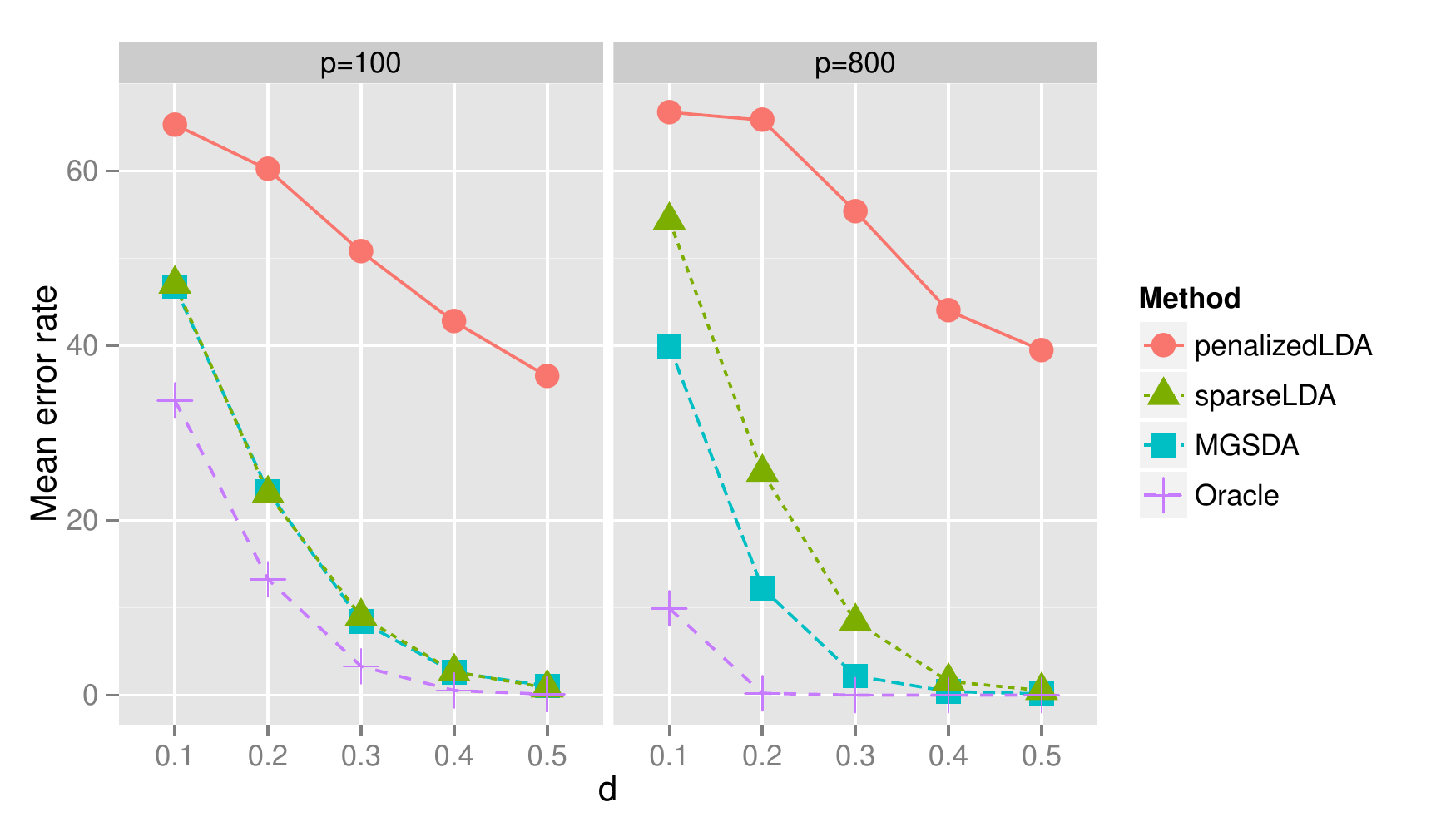}}
\end{figure}

\subsection{Implementation Details}\label{s:implementation}
The method of \citet{Cai:2011dm} is implemented using \textsf{linprogPD} function from the package \textsf{CLIME} from CRAN. Note that \textsf{linprogPD} almost never returns a sparse solution. However, all the values below the precision level should be treated as zeroes \citep{Cai:2011vi}. We used the default value of $10^{-3}$ for precision. The grid for the tuning parameter is chosen from 0.01 to 0.5 by 0.01. The method of \citet{Witten:2011kc} is implemented using the package \textsf{penalizedLDA} from CRAN. The grid for tuning parameter is chosen from 0 to 1 by 0.01. The method of \citet{Clemmensen:2011kr} is implemented using the package \textsf{sparseLDA} from CRAN. Each canonical vector is constrained to have between 3 and 80 features. This is quite a restrictive range for tuning, however the \textsf{sparseLDA} package produced errors when we used a wider range of features.  MGSDA is implemented using the R package \textsf{MGSDA}. The grid for the tuning parameter $\lambda_{1}\le...\le\lambda_{max}$ is chosen adaptively for each dataset with $\lambda_{max}=\max_{j}\|d_{j}\|_{2}$, which corresponds to zero selected features. For each $\lambda_{l}<\lambda_{\max}$ we set $V^{(0)}=\hat V(\lambda_{l+1})$. For all the methods, the final tuning parameter is chosen from the respective grid through 5-fold cross-validation to minimize the error rate.

Witten and Tibshirani's penalizedLDA has significantly faster running time than all other methods since penalizedLDA assumes that the covariance matrix has diagonal structure. This assumption results in a simplified optimization algorithm, for details we refer to \citet{Witten:2011kc}. The running time of penalizedLDA is followed by MGSDA and sparseLDA.  Surprisingly, LPD has the slowest performance. We suspect that this is not due to the method itself, but due to the use of \textsf{linprogPD} function in its implementation. A different linear program solver is likely to result in much faster running time, however the use of a general solver makes the method implementation less straightforward.

\section{Real Data}\label{sec:data}
\subsection{Metabolomics Dataset}
Metabolomics is the global study of all metabolites in a biological system under a given set of conditions.  Metabolites are the final products of enzymes and enzyme networks whose substrates and products often cannot be deduced from genetic information and whose levels reflect the integrated product of the genome, proteome and environment.  Metabolomic readouts thus represent the most direct (or phenotypic) readout of a cellÕs physiologic state.  From a technical standpoint, analytical studies of metabolism have been historically limited to one or a limited set of metabolites. However, advances in liquid chromatography and mass spectrometry have recently made it possible to measure hundreds of metabolites and with enough biomass well over 1000, in parallel.  Such technologies have thus opened the door to obtaining global biochemical readouts of a cellÕs physiologic state and response to perturbation.   Cornell researchers have developed and applied a state-of-the-art metabolomic platform to track the intrabacterial ÔpharmacokineticÕ fates and ÔpharmacodynamicÕ actions of a given compound within Mycobacterium tuberculosis \citep{ Pethe:2010fe, deCarvalho:2010dc, deCarvalho:2011eb, Chakraborty:2013hq}.  These studies demonstrate the highly unpredictable nature and identities of these properties even for well-studied antibiotics.  

We investigate a (currently unpublished) metabolomics data obtained from Dr. Kyu Rhee, which seeks to systematically elucidate the intrabacterial pharmacokinetic and pharmacodynamic fates and actions of antimycobacterial hit or lead compound series identified in high throughput screens against replicating and non- or slowly replicating forms of Mycobacterium tuberculosis. 

The data contains measurements of 171 metabolic responses of 68 patients to 25 antibiotics that are administered at different dosage levels. Each measurement is an average of three replicates, normalized to the vehicle control and log2 transformed. 14 out of 25 antibiotics can be divided into the following 5 groups: STREP\_AMI(strep, ami), FLQ(lev, moxi), DHFR(nitd2, sri8210, sri 8710, sri 8857), DHPS(smx, snl, aps) and InhA(eta,  isoxyl, gsk93). These antibiotics are administered to 35 patients out of 68. In the subsequent analysis we only focus on these 5 groups of antibiotics and do not consider the dosage levels.

We compare the performance of MGSDA, penalizedLDA \citep{Witten:2011kc} and sparseLDA \citep{Clemmensen:2011kr} on this dataset using the following measures: the mean number of misclassified samples and the mean number of selected features over 100 replications of 5-fold cross-validation. We do not perform random splits into the training and test set due to the small sample size. The results are reported in Table~\ref{t:metabolite}.

\begin{table}[!t]
\caption{Mean number of misclassified samples and mean number of selected features over 100 replications on metabolomics dataset, standard deviation is given in brackets.}
\label{t:metabolite}
\begin{center}
\begin{tabular}{|r|lll |}
  \hline
 & MGSDA & penalizedLDA & sparseLDA \\ 
  \hline
CV error& 0.074(0.11) & 0.072(0.11) & 0.014(0.05)  \\ 
 Features& 19(8)& 165(13) &34(12) \\ 
\hline
\end{tabular}
\end{center}
\end{table}
The results show that all three methods perform very well in terms of misclassification error; the mean number of misclassified samples is significantly less than one indicating that all three methods lead to almost perfect classification performance. Such a good performance suggests that there is a significant difference in the metabolic responses between the 5 groups of antibiotics.  However, penalizedLDA achieves this performance by selecting almost all of the metabolites, whereas MGSDA and sparseLDA use less than 20\% of the original features. Note that there is a substantial variation between the replications due to the small sample size of the data. 

\begin{figure}[!t]
\centering
\caption{Metabolomics dataset projected onto 4 column vectors of $V$, $k=8$ metabolite features are used.}
\label{fig:proj}
\makebox{\includegraphics[scale=0.65]{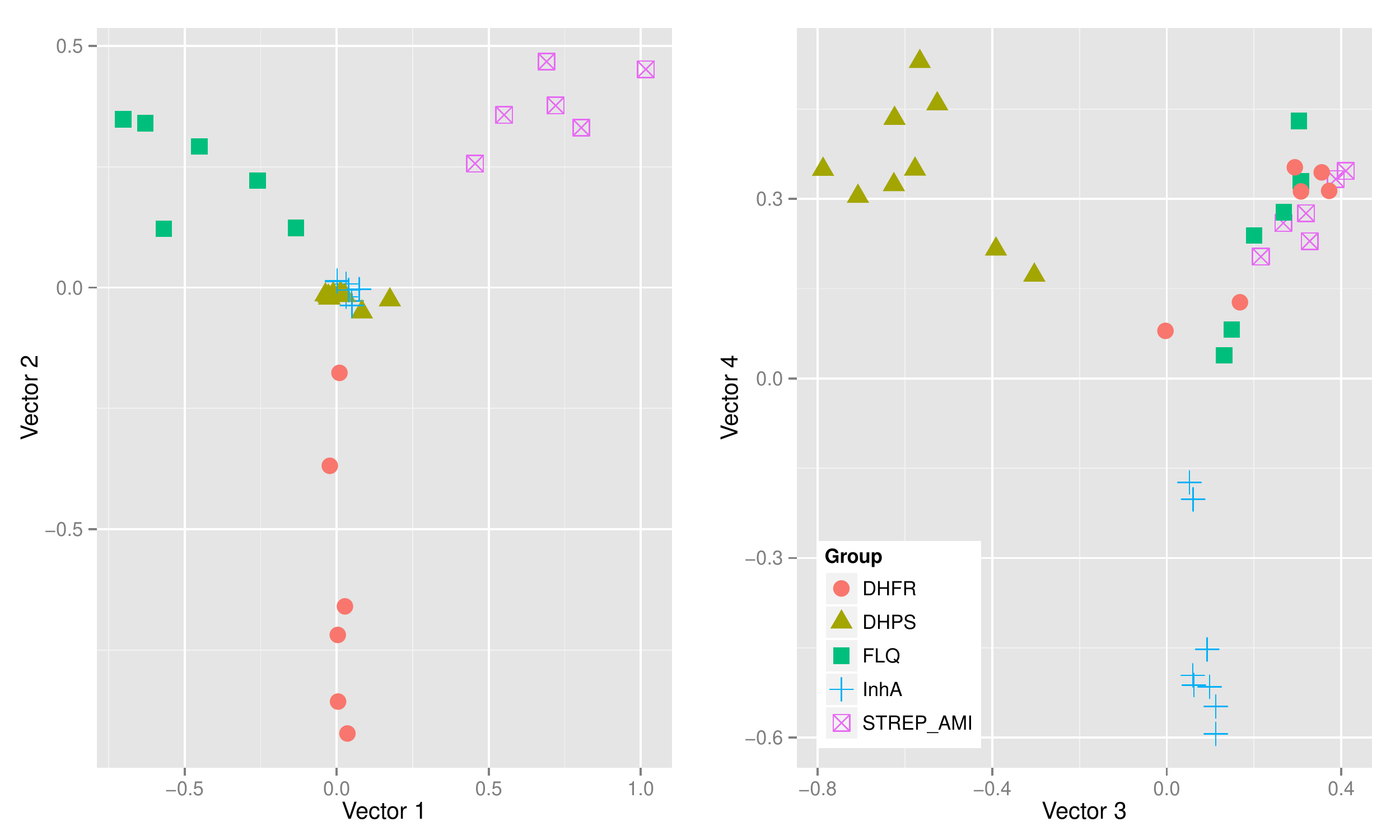}}
\end{figure}

We further estimate four canonical vectors using MGSDA with $\lambda=0.57$ and illustrate the projected data in Figure~\ref{fig:proj}. Note that 8 selected metabolites provide perfect linear separation between the groups. $\lambda=0.57$ is chosen as one of the hundred tuning parameters from above replications of cross-validation splits. We have tried the other values of $\lambda$ as well, however they all provided perfect linear separation between the groups with projected data being very similar to Figure~\ref{fig:proj}. Though there is a variation between the cross-validation replications due to the small sample size of the data, this variation has negligible effect on the final projection.


\subsection{14 Cancer Dataset} 

In this section we compare the performance of MGSDA, penalizedLDA \citep{Witten:2011kc} and sparseLDA \citep{Clemmensen:2011kr} on the 14 cancer dataset by \citet{Ramaswamy:2001hc}. This dataset contains 16063 gene expression measurements collected on 198 samples. Each sample belongs to one of the 14 cancer classes. The dataset can be obtained from \url{http://statweb.stanford.edu/~tibs/ElemStatLearn/}. We selected this dataset as it is publicly available and has been previously analyzed by a number of authors including \citet{Witten:2011kc}.

Following the recommendation of \citet[p.~654]{Hastie:2009fg}, we first standardize the data to have mean zero and standard deviation one for each patient. To reduce the overall computational cost, we restrict the analysis to 3000 genes. We select these genes following the novel model-free feature screening procedure for discriminant analysis of \citet{Cui:2014dg}. Following the approach taken by \citet{Witten:2011kc}, we perform 100 independent splits of the data set into the training set containing 75\% of the samples and the test set containing 25\% of the samples. The tuning parameter for all methods is selected using 5-fold cross-validation on the training set. The mean number of misclassified samples on the test set and the mean number of selected features over 100 splits are reported in Table~\ref{t:14cancer}. MGSDA and sparseLDA perform better than penalizedLDA in terms of the misclassification error and select much smaller number of features. MGSDA and sparseLDA select comparable number of features, with the misclassification error of MGSDA being the smallest.
\begin{table}[!t]
\caption{Mean number of misclassified samples and mean number of selected features over 100 splits on 14 cancer dataset, standard deviation is given in brackets.}
\label{t:14cancer}
\begin{center}
\begin{tabular}{|r|lll|}
  \hline
 & MGSDA & penalizedLDA & sparseLDA \\ 
  \hline
  Error & 7.76(2.35) & 13.41(2.49) & 9.29(2.40) \\ 
  Features & 295(78) & 2962(34) & 293(55) \\ 
   \hline
\end{tabular}
\end{center}
\end{table}

\section{Discussion}\label{sec:disc}
This paper introduces a novel procedure that estimates population canonical vectors in the multi-group setting, and a corresponding R package MGSDA is available on CRAN. The proposed method is a natural generalization of the two-group methods that were previously studied in the literature. In addition to being computationally tractable, the method performs feature selection which results in sparse canonical vectors. The group penalty eliminates features from all canonical vectors at once with the remaining non-zero features being the same for all the vectors.

One possible extension of the proposed method is to allow canonical vectors to have different sparsity patterns. This goal can be achieved through the addition of the within-row penalty term to the objective function \eqref{eq:Frob2}. Such an estimation procedure has already been considered in the regression context; for example, \citet{Simon:2012tg} propose the following optimization problem:
 \begin{equation*}
  \hat \beta(\lambda,\alpha) = \arg \min_{\beta \in \mathbb{R}^p} \frac1{2n}\|Y-X\beta\|^2_2+(1-\alpha)\lambda\sum_{g=1}^{G}\sqrt{p_g}\|\beta^{(g)}\|_2+\alpha\lambda\|\beta\|_1,
 \end{equation*}
 where $p_g$ is the size of group $g$. In our case this approach results in
 \begin{equation*}
 \hat V(\lambda,\alpha)=\arg \min_{V\in\R^{p\times(G-1)} }\frac12\Tr\left(V^tW V\right)+\frac12\|D^tV-I\|^2_F+(1-\alpha) \lambda\sum_{i=1}^p\|v_i\|_2+\alpha \lambda\|V\|_{1}.
 \end{equation*}
 
Another possible extension is to perform canonical vectors selection in addition to feature selection, which will enhance the interpretability, especially when the number of groups $G$ is large.  This goal can be achieved through the addition of the nuclear norm penalty term to the objective function \eqref{eq:Frob2}:
 \begin{equation*}
 \hat V(\lambda,\alpha)=\arg \min_{V\in\R^{p\times(G-1)}}\frac12\Tr\left(V^tW V\right)+\frac12\|D^tV-I\|^2_F+\lambda \sum_{i=1}^p\|v_i\|_2+\alpha \|V\|_{*}.
 \end{equation*}
 Depending on the value of $\alpha>0$, the resulting matrix $\hat V$ has rank that is less than $G-1$, effectively resulting in a lower-dimensional eigenspace.
 
Both extensions result in convex optimization problems, but require additional modifications to the optimization algorithm. An interesting direction for future research is to examine how these extensions compare to the original method in different scenarios.


We established the variable selection and classification consistency of proposed estimator in the regime where $\frac{\log(ps)s^{2}}{N}\to 0$. While preparing this manuscript, we became aware of the work of \citet{Kolar:2013vs}, who show variable selection consistency of the sparse discriminant analysis under the conditions that $G=2$ and $N\ge Cs\log((p-s)\log(N))$ for some constant $C>0$. These improved rates directly apply to our proposal in the case $G=2$,  however the extension of these results to the case $G>2$ is not clear.  This is another direction for future research.



\section*{Appendix A}
\label{sec:Appendix}
\subsection*{Proof of Proposition \ref{cl:BD}}
\begin{proof} The proof is only given for matrix $B$, the proof for matrix $\varSigma_B$ is similar. 

1. Consider the equal group case: $n_1=...=n_G=n$ and $N=Gn$. It follows that $\bar X=\sum_{i=1}^G\bar X_G /G$ and therefore $B=\frac1{N}\sum_{g=1}^Gn(\bar X_g -\bar X)(\bar X_g -\bar X)^t=\frac1{N}X^t\{\frac1{\sqrt{n}}\mathsf{1}_g\}C \{\frac1{\sqrt{n}}\mathsf{1}_g\}^tX$, where $C$ is the centering matrix and $\{\frac1{\sqrt{n}}\mathsf{1}_g\}$ is a $N\times G$ matrix formed by $G$ columns $\frac1{\sqrt{n}}\mathsf{1}_g$ such that $(\mathsf{1}_g)_j=1$ if $j$th observation belongs to the $g$th group and $(\mathsf{1}_g)_j=0$ otherwise. Note that $C=H^tH$ where $H$ is the Helmert matrix of size $G$ with its first row removed \citep{Searle:2006ww}. Therefore $B=DD^{t}$, where $D=\frac1{\sqrt{N}}X^t\{\frac1{\sqrt{n}}\mathsf{1}_g\}H^t$. 

2. Consider the general case where each group has size $n_g$. Similar to the equal group case, $B=\frac1{N}X^t\{\frac1{\sqrt{n_i}}\mathsf{1}_i\}\tilde C \{\frac1{\sqrt{n_i}}\mathsf{1}_i\}^tX$, where $\tilde C=I_G-\frac1{\sqrt N}KK^t$ and $K=(\sqrt{n_1} ... \sqrt{n_G})^t$. Next we show that similar to $C$, $\tilde C$ can be decomposed as $\tilde C=\tilde H^t\tilde H$ and $\tilde H$ is a $G-1 \times G$ adjusted Helmert matrix. Since $\tilde H$ satisfies $I_G-\frac1{\sqrt N}KK^t=\tilde H^t\tilde H$, $(K, \tilde H^t)$ is an orthogonal matrix. The $G-1$ orthogonal contrasts for unbalanced data have the following form \citep[p~51]{Searle:2006ww}: 
\begin{equation*}
  \delta_r=\sqrt{n_{r+1}}\left(\sum_{h=1}^{r}n_h(\bar X_h - \bar X_{r+1})\right).
\end{equation*}
Denote by $h_r$ the rows of $\tilde H$. Then it follows that for some constant $C_{r}$, $h_r\left\{\frac1{\sqrt{n_i}}\mathsf{1}_i\right\}^tX=C_r\delta_r.$
This means that $h_{rj}=C_r\sqrt{n_{r+1} n_j}$ for $j=1,..,r$; $h_{r(r+1)}=-C_r\sum_{i=1}^{r}n_{i}$ and $h_{rj}=0$ for $j>(r+1)$. To find $C_r$, we use the fact that $h_rh_r^t=1$. Let $s_{r}=\sum_{i=1}^{r}n_{i}.$ Then $C_r$ satisfies
$C_r^2\left(\sum_{j=1}^{r}n_{r+1} n_j + s_{r}^2\right)=1$, or equivalently $C_r^2 s_{r+1}s_{r}=1$.
From the last equation $C_r=\frac{1}{\sqrt{s_{r+1} s_{r}}}$. Combining the results it follows that $B=DD^t,$ where $D=\frac1{\sqrt{N}}X^t\left\{\frac1{\sqrt{n_g}}\mathsf{1}_i\right\}\tilde H^t$ and $D_{r}=\frac1{\sqrt{N}}C_{r}\delta_{r}=\frac{\sqrt{n_{r+1}}\left(\sum_{h=1}^{r}n_h(\bar X_h - \bar X_{r+1})\right)}{\sqrt{Ns_{r+1} s_{r}}}$.
\end{proof}

\subsection*{Proof of Proposition \ref{cl:V}}
\begin{proof}
Denote $\varPsi=\varSigma_W^{-1}\Delta \Rho$, where $\Rho$ is an orthogonal matrix such that $\Delta^t\varSigma_W^{-1}\Delta=\Rho\Lambda \Rho^t$. It follows that $\varSigma_W^{-1}\varSigma_B\varPsi=\varSigma_W^{-1}\Delta \Delta^t\varSigma_W^{-1}\Delta \Rho=\varSigma_W^{-1}\Delta \Rho\varLambda=\varPsi\varLambda$.  Hence, $\varPsi$ is the matrix of eigenvectors of $\varSigma_{W}^{-1}\varSigma_{B}$. The proof for $V$ is analogous.
\end{proof}

\subsection*{Proof of Proposition \ref{cl:clas}}
\begin{proof} The proof is only given for the sample classification rule $\hat h_{V}(x)$, the proof for the population classification rule $h_{\varPsi}(x)$ is analogous.
Define $Z=XV$. Using $V$, a new observation $x\in \R^p$ is classified to group $\hat h_{V}(x)$, where
  \begin{equation*}
 \hat h_{V}(x)=\arg \min_{1\le j\le G}(V^tx-\bar Z_j)^t(V^tWV)^{-1}(V^tx-\bar Z_j)-2\log\frac{n_j}{N}.
  \end{equation*} 
Consider a new classification rule $\hat h_{V'}(x)$ based on $V'=VR$ with $R\in\mathbb{O}^{G-1}$. Then $$Z'=XV'=XVR=ZR$$and
  \begin{equation*}
  \begin{split}
 \hat h_{V'}(x)&=\arg \min_{1\le j\le G}(V'^tx-\bar Z'_j)^t(V'^tWV')^{-1}(V'^tx-\bar Z'_j)-2\log\frac{n_j}{N}\\
 &=\arg \min_{1\le j\le G}(R^tV^tx-R^t\bar Z_j)^t(R^tV^tWVR)^{-1}(R^tV^tx-R^t\bar Z_j)-2\log\frac{n_j}{N}\\
 &=\arg \min_{1\le j\le G}(V^tx-\bar Z_j)^tRR^{-1}(V^tWV)^{-1}(R^t)^{-1}R^t(V^tx-\bar Z_j)-2\log\frac{n_j}{N}\\
 &=\hat h_{V}(x).
 \end{split}
  \end{equation*} 
\end{proof}
\subsection*{Proof of Proposition \ref{cl:eigen}}
\begin{proof}
 From the definition of $\varUpsilon_{\rho}$,  $(\varSigma_W+\rho\varSigma_B)^{-1}\rho\varSigma_B\varUpsilon_{\rho}=\varUpsilon_{\rho}\varLambda$. It follows that 
 \begin{equation*}
 \begin{split}
  &\rho\varSigma_B\varUpsilon_{\rho}=\varSigma_W\varUpsilon_{\rho}\varLambda+\rho\varSigma_B\varUpsilon_{\rho}\varLambda;\\
  &\rho\varSigma_B\varUpsilon_{\rho}(I-\varLambda)=\varSigma_W\varUpsilon_{\rho}\varLambda;\\
  &\varSigma_W^{-1}\varSigma_B\varUpsilon_{\rho}=\varUpsilon_{\rho}\frac1{\rho}\varLambda(I-\varLambda)^{-1}.
 \end{split}
 \end{equation*}
 From the last equation it follows that $\varUpsilon_{\rho}$ is the matrix of eigenvectors of $\varSigma_W^{-1}\varSigma_B$.  Since the eigenvectors are unique up to normalization, the statement of the proposition follows.
\end{proof}

\subsection*{Proof of Proposition \ref{cl:Mai}}
\begin{proof}
By definition $D=\frac{\sqrt{n_1n_2}}{N}(\bar X_1-\bar X_2)$. Therefore
\begin{equation*}
\begin{split}
\hat V_{DSDA}(\lambda)&=\arg\min_{V\in\R^{p}} \frac12V^t(W+DD^t)V-\frac{N}{\sqrt{n_1n_2}}D^tV+\lambda\|V\|_1\\
&=\arg\min_{V\in\R^{p}} \frac 12\frac{\sqrt{n_1n_2}}{N}  V^t(W+DD^t)V-D^tV+\frac{\lambda\sqrt{n_1n_2}}{N}\|V\|_1.
\end{split}
\end{equation*}
Define
\begin{equation*}
 f_{DSDA}(V,\lambda)=\frac 12\frac{\sqrt{n_1n_2}}{N} V^t(W+DD^t)V-D^tV+\frac{\lambda\sqrt{n_1n_2}}{N}\|V\|_1.
\end{equation*}
Similarly, $\hat V(\lambda)=\arg \min_{V\in \R^{p}} f(V,\lambda)$, where
\begin{equation*}
 f(V,\lambda)=\frac12V^t(W+DD^t)V-D^tV+\lambda\|V\|_1.
\end{equation*}
Note that
\begin{equation*}
\begin{split}
 f\left(\frac{\sqrt{n_1n_2}}{N}V,\lambda\right)&=\frac{\sqrt{n_1n_2}}{N}\left(\frac12\frac{\sqrt{n_1n_2}}{N}V^t(W+DD^t)V-D^tV+\lambda\|V\|_1\right)\\
 &=\frac{\sqrt{n_1n_2}}{N} f_{DSDA}\left(V,\frac{N}{\sqrt{n_1n_2}}\lambda\right).
 \end{split}
\end{equation*}
It follows that $\hat V(\lambda)=\frac{N}{\sqrt{n_1n_2}}\hat V_{DSDA}\left(\frac{N}{\sqrt{n_1n_2}}\lambda\right)$. 
\end{proof}
\subsection*{Auxillary lemmas for Theorem~\ref{t:main1}}
\begin{lemma}\label{l:inf2bound}
 $\|AB\|_{\infty,2}\le \|A\|_{\infty}\|B\|_{\infty,2}.$
\end{lemma}
\begin{proof}
This inequality is a special case of Lemma 8 in \citet{Obozinski:2011ho}. Note that 
\begin{equation*}
\|A\|_{\infty,2}=\max_i\|a_i\|_2=\max_i\max_{\|y_i\|_2\le 1}|y^ta_i|=\max_{\|y\|_2\le 1}\max_i|y^ta_i|=\max_{\|y\|_2\le 1}\|Ay\|_{\infty}.
\end{equation*}
It follows that 
\begin{equation*}
 \|AB\|_{\infty,2}=\max_{\|y\|_2\le 1}\|ABy\|_{\infty}\le\max_{\|y\|_2\le 1}\|A\|_{\infty}\|By\|_{\infty}=\|A\|_{\infty}\max_{\|y\|_2\le 1}\|By\|_{\infty}=\|A\|_{\infty}\|B\|_{\infty,2}.
\end{equation*}
\end{proof}

\begin{lemma}\label{l:meand}
Let $F=D-\varDelta$. There exists constant $c_{3}>0$ such that
 \begin{equation*}
  P\left(\|F\|_{\infty,2}\ge \epsilon\right)\le 2p(G-1)\exp(-c_3N\epsilon^2).
 \end{equation*}
\end{lemma}
\begin{proof}
From the definition of $\varDelta$ and under the assumption $\pi_{g} = 1/G$, its $r$th column has the form 
$ \varDelta_{r}=\sum_{i=1}^{r}(\mu_i-\mu_{r+1})/\sqrt{Gr(r+1)}.$ Similarly, $D_{r}=\sum_{i=1}^{r}(\bar X_i-\bar X_{r+1})/\sqrt{Gr(r+1)}$. Therefore,
\begin{equation*}
 F_{r}=\frac{1}{\sqrt{Gr(r+1)}}\sum_{i=1}^{r}\big((\bar X_i-\bar X_{r+1})-(\mu_i-\mu_{r+1})\big).
\end{equation*}
Since the groups are independent and $(\bar X_{g})_{j}\sim N\left((\mu_g)_{j},\frac{\sigma^2_j}{n}\right)$ for all $g\in\{1,...,G\}$ and $j\in\{1,...,p\}$, then for all $r$:
 $$\sum_{i=1}^{r}(\bar X_i-\bar X_{r+1})_{j}\sim N\left(\sum_{i=1}^{r}\big(\mu_i-\mu_{r+1})_{j},\frac{r(r+1)\sigma^2_j}{n}\right),$$
 or equivalently
 $$d_{jr}\sim N\left(\delta_{jr},\frac{\sigma^2_j}{Gn}\right),$$
 where $d_{jr}$ are the elements of matrix $D$ and $\delta_{jr}$ are the elements of matrix $\Delta$.
It follows that for all $r\in \{1,...,G-1\}$ and for all $j\in\{1,...,p\}$
\begin{equation*}
P(|f_{jr}|\ge \epsilon)=P(|d_{jr}-\delta_{jr}|\ge \epsilon)\le 2\exp\left(-\frac{N\epsilon^2}{2\sigma^2_j}\right)\le2\exp(-cN\epsilon^{2}).
\end{equation*}
Therefore,
\begin{align*}
P(\|f_j\|_2\ge \epsilon)&=P\left(\sqrt{f_{j1}^2+...+f_{(G-1)j}^2}\ge \epsilon\right)\le P\left(\sqrt{G-1}\max_{r}|f_{jr}|\ge \epsilon\right)\\
&\le P\left(\cup_{r}\left\{|f_{jr}|\ge \frac{\epsilon}{\sqrt{G-1}}\right\}\right)\le (G-1)P\left(|f_{jr}|\ge \frac{\epsilon}{\sqrt{G-1}}\right)\\
&\le 2(G-1)\exp(-c_{3}N\epsilon^{2}).
\end{align*}
The result follows by applying the union bound over $j\in\{1,...,p\}$. 
\end{proof}
\begin{lemma} \label{l:T}
Let $T=W+B$ and $\varSigma=\varSigma_{W}+\varSigma_{B}$. There exist constants $c_{1}>0$ and $c_{2}>0$ such that
\begin{equation*}
\begin{split}
 &P\left(\|T_{AA}-\varSigma_{AA}\|_{\infty}\ge \epsilon\right)\le c_{1}s^2\exp(-c_{2}Ns^{-2}\epsilon^2);\\
 &P\left(\|T_{A^cA}-\varSigma_{A^cA}\|_{\infty}\ge \epsilon\right)\le c_{1}s(p-s)\exp(-c_{2}Ns^{-2}\epsilon^2).\\
\end{split}
\end{equation*}
\end{lemma}
\begin{proof}
 First, we show that $P(|\varSigma_{ij}-T_{ij}|>\epsilon)]\le c_{1}\exp(-c_{2}N\epsilon^2).$
 By definition,
\begin{equation*}
\begin{split}
\varSigma_{ij}-T_{ij}&=\varSigma_{W_{ij}}+\varSigma_{B_{ij}}-\frac1{N}\sum_{k=1}^N(X_{ki}-\bar X_{i})(X_{kj}-\bar X_{j})\\
&=\varSigma_{W_{ij}}+\sum_{g=1}^G\pi_g(\mu_{gi}-\mu_{i})(\mu_{gj}-\mu_{j})-\frac{1}{N}\sum_{k=1}^NX_{ki}X_{kj}+\bar X_{i}\bar X_{j}\\
&=\varSigma_{W_{ij}}+\sum_{g=1}^G\pi_g\mu_{gi}\mu_{gj}+\bar X_{i}\bar X_{j}-\mu_i\mu_j-\frac{1}{N}\sum_{g=1}^G\sum_{k \in I_g}X_{ki}X_{kj}.
\end{split}
\end{equation*}
Furthermore,
\begin{equation*}
\frac1{n_g}\sum_{k\in I_g}X_{ki}X_{kj}=\frac1{n_g}\sum_{k\in I_g}(X_{ki}-\mu_{gi})(X_{kj}-\mu_{gj})+\mu_{gi}(\bar X_{gj}-\mu_{gj})+\mu_{gj}(\bar X_{gi}-\mu_{gi})+\mu_{gj}\mu_{gi}.
\end{equation*}
Therefore
\begin{equation*}
\begin{split}
\varSigma_{ij}-T_{ij}=&\varSigma_{W_{ij}}+\sum_{g=1}^G\pi_g\mu_{gi}\mu_{gj}+\bar X_{i}\bar X_{j}-\mu_i\mu_j-\\
&-\frac{1}{N}\sum_{g=1}^Gn_g\left(\frac1{n_g}\sum_{k\in I_g}( X_{ki}-\mu_{gi})(X_{kj}-\mu_{gj})+\mu_{gi}(\bar X_{gj}-\mu_{gj})+\mu_{gj}(\bar X_{gi}-\mu_{gi})+\mu_{gj}\mu_{gi}\right)\\
=&\sum_{g=1}^G\frac{n_g}{N}\left(\varSigma_{W_{ij}}-\frac1{n_g}\sum_{k\in I_g}( X_{ki}-\mu_{gi})(X_{kj}-\mu_{gj})\right)\\
&+\sum_{g=1}^G\frac{n_g}{N}\left(\mu_{gi}(\mu_{gj}-\bar X_{gj})+\mu_{gj}(\mu_{gi}-\bar X_{gi})\right)+\sum_{g=1}^G\left(\pi_g-\frac{n_g}{N}\right)\mu_{gi}\mu_{gj}+(\bar X_{i}\bar X_{j}-\mu_i\mu_j)
\end{split}
\end{equation*}
Under the assumption $\pi_{g}=\frac1{G}$ and $n_{g}=\frac1{G}$, the above expression is further simplified as
\begin{align*}
\varSigma_{ij}-T_{ij}=&\frac{1}{G}\sum_{g=1}^G\left(\varSigma_{W_{ij}}-\frac1{n_g}\sum_{k\in I_g}( X_{ki}-\mu_{gi})(X_{kj}-\mu_{gj})\right)\\
&+\frac{1}{G}\sum_{g=1}^G\left(\mu_{gi}(\mu_{gj}-\bar X_{gj})+\mu_{gj}(\mu_{gi}-\bar X_{gi})\right)+(\bar X_{i}\bar X_{j}-\mu_i\mu_j)\\
=&I_1+I_2+I_3
\end{align*}
For the final bound it remains to show that  for each $I_j$  there exist constants $c_{1j}>0$ and $c_{2j}>0$ such that $P(|I_j|\ge\epsilon)\le c_{1j}\exp(-c_{2j}N\epsilon^2).$

\textbf{Analysis of $I_{1}$.} From Lemma A.3 in \citet{Bickel:2008wv}, there exist constants $C_{1}>0$ and $C_{2}>0$ such that for $\epsilon<\epsilon_{0}$
\begin{equation*}
P\left(\left|\frac1{n}\sum_{k=1}^{n}(Z_{ik}Z_{jk}-\sigma_{jk})\right|\ge \epsilon\right)\le C_{1}\exp(-C_{2}n\epsilon^{2}),
\end{equation*}
where $Z_{i}$ are i.i.d $N(0, \varSigma)$ and $\sigma_{ij}$ are elements of $\varSigma$.
Let $\tilde Z_{i}=X_{i}-\mu_{i}$, where $\mu_{i}=\mu_{g}$ if observation $i$ belongs to group $g$. By definition $\E(\tilde Z_{i})=0$ and $\tilde Z_{i}$ are i.i.d $N(0,\varSigma_{W})$. Note that $I_{1}$ can be rewritten as
$$I_1=\frac1{N}\sum_{g=1}^G\sum_{k\in I_g}\left(\varSigma_{W_{ij}}-(X_{ki}-\mu_{gi})(X_{kj}-\mu_{gj})\right)=\frac1{N}\sum_{l=1}^{N}(\varSigma_{W_{ij}}-\tilde Z_{li}\tilde Z_{lj}).$$
Therefore $$P(|I_1|\ge \epsilon)\le P\left(\left|\frac1{N}\sum_{l=1}^{N}(\varSigma_{W_{ij}}-\tilde Z_{li}\tilde Z_{lj})\right|\ge \epsilon\right)\le C_{1}\exp(-C_{2}N\epsilon^2).$$
\textbf{Analysis of $I_{2}$.} Dy definition, $I_2=\frac1{G}\sum_{g=1}^G\left(\mu_{gi}(\mu_{gj}-\bar X_{gj})+\mu_{gj}(\mu_{gi}-\bar X_{gi})\right)$. It follows that
\begin{equation*}
|I_2|\le 2\max\left(\left|\frac1{G}\sum_{g=1}^{G}\mu_{gi}(\mu_{gj}-\bar X_{gj})\right|,\left|\frac1{G}\sum_{g=1}^{G}\mu_{gj}(\mu_{gi}-\bar X_{gi})\right|\right).
\end{equation*}
Since the groups are independent, $$\frac1{G}\sum_{g=1}^{G}\mu_{gi}(\mu_{gj}-\bar X_{gj})\sim N\left(0,\frac{\sum_{g=1}^{G}\mu_{gi}^{2}}{NG}\sigma^{2}_{j}\right).$$
Therefore,
\begin{equation*}
P\left(\left|\frac1{G}\sum_{g=1}^{G}\mu_{gi}(\mu_{gj}-\bar X_{gj})\right|\ge \epsilon\right)\le 2\exp(-cN\epsilon^2),
\end{equation*}
hence
\begin{equation*}
P(|I_2|\ge \epsilon)\le 4\exp\left(-\frac{c}{4}N\epsilon^2\right).
\end{equation*}
\textbf{Analysis of $I_{3}$.} Note that
\begin{equation*}
\begin{split}
I_{3}=\bar X_{i}\bar X_{j}-\mu_i\mu_j=&\sum_{g=1}^G\frac{1}{G}\bar X_{gi}\sum_{l=1}^G\frac{1}{G}\bar X_{lj}-\sum_{g=1}^G\frac1{G}\mu_{gi}\sum_{l=1}^G\frac1{G}\mu_{lj}\\
=&\sum_{g=1}^G\frac{1}{G}\left(\bar X_{gi}-\mu_{gi}\right)\sum_{l=1}^G\frac{1}{G}\left(\bar X_{lj}-\mu_{lj}\right)\\
&+\mu_{i}\sum_{l=1}^G\frac{1}{G}(\bar X_{lj}-\mu_{lj})+\mu_{j}\sum_{g=1}^G\frac{1}{G}(\bar X_{gi}-\mu_{gi}).
\end{split}
\end{equation*}
Therefore,
\begin{equation*}
|I_3|\le \max_{t \in\{i,j\}}\left|\frac1{G}\sum_{g=1}^{G}(\bar X_{gt}-\mu_{gt})\right|^2+2\max_g{|\mu_{g}|}\max_{t\in\{i,j\}}\left|\frac1{G}\sum_{g=1}^{G}(\bar X_{gt}-\mu_{gt})\right|.
\end{equation*}
Since the groups are independent, 
 $$\frac1{G}\sum_{g=1}^{G}\bar X_{gj}\sim N\left(\frac1{G}\sum_{g=1}^{G}\mu_{gj},\frac{\sigma^2_j}{nG}\right).$$
Therefore,
\begin{equation*}
P\left(\max_{t\in\{i,j\}}\left|\frac1{G}\sum_{g=1}^{G}(\bar X_{gt}-\mu_{gt})\right|\ge \epsilon\right)\le 2\cdot2\exp(-c_{1}N\epsilon^2).
\end{equation*}
Note that if $\max_{t\in\{i,j\}}\left|\frac1{G}\sum_{g=1}^{G}(\bar X_{gt}-\mu_{gt})\right|\le k\epsilon$, then $|I_{3}|\le k^{2}\epsilon^{2}+2k\epsilon\max_g{|\mu_{g}|}$. Choosing $k\le\frac{1}{\max(\sqrt 2,2\max_{g}\|\mu_{g}\|)}$ leads to
$|I_{3}|\le \frac{\epsilon^{2}+\epsilon}{2}\le \epsilon$ for small values of $\epsilon$. Hence,
\begin{equation*}
P(|I_3|\ge\epsilon)\le 4\exp(-cN\epsilon^2).
\end{equation*}
Combining the results for $I_{1}$-$I_{3}$ leads to $c_{1}=C_{1}+4+4$. The results of lemma follow from the definition of $\|\cdot\|_{\infty}$ and the union bound.
\end{proof}

\begin{lemma}\label{lemma:Tw} Let $F_T=T_{A^cA}(T_{AA})^{-1}-\varSigma_{A^cA}(\varSigma_{AA})^{-1}$. There exists constant $c_{3}>0$ such that 
 \begin{equation*}
  P\left(\|F_T\|_{\infty}\ge \epsilon\phi(\kappa+1)(1-\phi\epsilon)^{-1}\right)\le 6(G+1)ps\exp(-c_3Ns^{-2}\epsilon^2).
 \end{equation*}
\end{lemma}
\begin{proof}
The proof follows the proof of Lemma A2 in \citet{Mai:2012bf} and uses the results of Lemma \ref{l:T}.
\end{proof}

\subsection*{Proof of Theorem 1.}
The proof follows the proof of Theorem 1 in \citet{Mai:2012bf} using the results of auxillary lemmas.

\subsection*{Proof of Corollary~1}
\begin{proof}
Follows directly from parts 1 and 2 of Theorem~1.
\end{proof}
\subsection*{Proof of Corollary~2}
\begin{proof}
The first result follows directly from part 3 of Theorem~1. To show the second result, we consider the events
\begin{align*}
\mathcal{E}_{1}&=\cap_{g}\left\{\left|\log \pi_{g}-\log\frac{n_{i}}{N}\right|\le C_{1}\frac1{\sqrt{N}}\right\};\\
\mathcal{E}_{2}&=\left\{\|\hat V_{A}-\varPsi_{A}\|_{\infty,2}\le C_{2}\sqrt{\frac{\log(ps)s^{2}}{N}}\right\};\\
\mathcal{E}_{3}&=\left\{\|W_{A}-\varSigma_{WAA}\|_{\infty}\le C_{3}\sqrt{\frac{\log(s^{2})s^{2}}{N}}\right\};\\
\mathcal{E}_{4}&=\cap_{g}\left\{\|\mu_{gA}-\bar x_{gA}\|_{\infty}\le C_{4}\sqrt{\frac{\log(s)}{N}}\right\};\\
\mathcal{E}_{5}&=\{A=\hat A\},
\end{align*}
where $C_{i}$ are constants independent of $n$, $p$ and $s$ and let $\mathcal{E}=\cap_{i} \mathcal{E}_{i}$. Given a new observation $X\in \R^{p}$ with a value $x$, define for each $g\in \{1,...,G\}$
\begin{align*}
h^{g}=h^{g}(x)&=(x-\mu_{g})^{t}\varPsi'(\varPsi'^{t}\varSigma_{W}\varPsi')^{-1}\varPsi'^{t}(x-\mu_{g})-2\log \pi_{g};\\
\hat h^{g}=\hat h^{g}(x)&=(x-\bar x_{g})^{t}\hat V(\hat V^{t}W\hat V)^{-1}\hat V^{t}(x-\bar x_{g})-2\log \frac{n_{g}}{N}.
\end{align*}
Since the classification rule is invariant to scaling and orthogonal rotation, it follows that population classification rule $h_{\varPsi}(x)=\arg\min_{g}h^{g}(x)$ and the sample classification rule $\hat h_{\hat V}(x)=\arg\min_{g}\hat h^{g}(x)$.
We first prove $\hat h^{g}\overset{P}{\longrightarrow}h^{g}$: there exists constant $C$ such that
on $\mathcal{E}$
  \begin{equation*}
|h^{g}-\hat h^{g}|\le  C\sqrt{\frac{\log(ps)s^{2}}{N}},
\end{equation*}
and $P(\mathcal{E})\to 1$ under (C1) and (C2). Let $a=\varPsi_{A}^{t}(x_{A}-\mu_{gA})$, $\hat a=\hat V_{A}^{t}(x_{A}-\bar x_{gA})$, $\varLambda^{-1}=(\varPsi_{A}\Sigma_{WAA}\varPsi_{A})^{-1}$ and $\hat \varLambda^{-1}=(\hat V_{A}^{t}W_{AA}\hat V_{A})^{-1}$. Consider
\begin{align*}
|h^{g}-\hat h^{g}|&=|a^{t}\varLambda^{-1} a-\hat a\hat \varLambda^{-1}\hat a|\\
&=|(\hat a-a)^{t}(\hat \varLambda^{-1}-\varLambda^{-1})(\hat a-a)+2a^{t}(\hat \varLambda^{-1}-\varLambda^{-1})(\hat a-a)+2a^{t}\varLambda^{-1}(\hat a-a)+a^{t}(\hat \varLambda^{-1}-\varLambda^{-1})a|\\
&\le (\|\hat a-a\|^{2}_{2}+2\|a\|_{2}\|_{2}\|\hat a-a\|_{2}+\|a\|_{2})\|\hat \varLambda^{-1}-\varLambda^{-1}\|_{2}+2\|a\|_{2}\|\varLambda^{-1}\|_{2}\|\hat a-a\|_{2}.
\end{align*}
By definition of $\hat a$ and $a$, on $\mathcal{E}$
\begin{align*}
\|\hat a-a\|_{2}&=\|\varPsi_{A}^{t}(x_{A}-\mu_{gA})-\hat V_{A}^{t}(x_{A}-\bar x_{gA})\|_{2}\\
&\le \|x_{A}^{t}(\varPsi_{A}-\hat V_{A})\|_{2}+\|(\mu_{gA}-\bar x_{gA})^{t}(\varPsi_{A}-\hat V_{A})\|_{2}\\
&\le \|\varPsi_{A}-\hat V_{A}\|_{\infty,2}(\|x_{A}\|_{\infty}+\|\mu_{gA}-\bar x_{gA}\|_{\infty}).
\end{align*}
Therefore, there exists constant $C'$ such that on the event $\mathcal{E}$
\begin{equation*}
\|\hat a-a\|_{2}\le C'\sqrt{\frac{\log(ps)s^{2}}{N}}.
\end{equation*}
By definition of $\hat \varLambda$ and $\varLambda$, $\hat \varLambda=\varLambda+\Upsilon$,
where 
\begin{align*}
\Upsilon&=(\hat V_{A}-\varPsi_{A})^{t}(W_{AA}-\varSigma_{WAA})(\hat V_{A}-\varPsi_{A})+\varPsi_{A}^{t}(W_{AA}-\varSigma_{WAA})(\hat V_{A}-\varPsi_{A})+(\hat V_{A}-\varPsi_{A})^{t}\varSigma_{WAA}(\hat V_{A}-\varPsi_{A})\\
&+\varPsi_{A}^{t}\varSigma_{WAA}(\hat V_{A}-\varPsi_{A})+(\hat V_{A}-\varPsi_{A})^{t}(W_{AA}-\varSigma_{WAA})\varPsi_{A}+\varPsi_{A}^{t}(W_{AA}-\varSigma_{WAA})\varPsi+(\hat V_{A}-\varPsi_{A})^{t}\varSigma_{WAA}\varPsi_{A}.
\end{align*}
Moreover,
\begin{equation*}
\|\hat \varLambda^{-1}-\varLambda^{-1}\|_{2}\le \|\hat \varLambda^{-1}\|_{2}\|\hat \varLambda-\varLambda\|_{2}\|\varLambda^{-1}\|_{2}\le (\|\varLambda^{-1}\|_{2}+\|\hat \varLambda^{-1}-\varLambda^{-1}\|_{2})\|\Upsilon\|_{2}\|\varLambda^{-1}\|_{2}.
\end{equation*}
By triangle inequality
\small
\begin{align*}
\|\Upsilon\|_{2}&\le\|(\hat V_{A}-\varPsi_{A})^{t}(W_{AA}-\varSigma_{WAA})(\hat V_{A}-\varPsi_{A})\|_{2}+2\|\varPsi_{A}^{t}(W_{AA}-\varSigma_{WAA})(\hat V_{A}-\varPsi_{A})\|_{2}\\
&+\|(\hat V_{A}-\varPsi_{A})^{t}\varSigma_{WAA}(\hat V_{A}-\varPsi_{A})\|_{2}+2\|\varPsi_{A}^{t}\varSigma_{WAA}(\hat V_{A}-\varPsi_{A})\|_{2}+\|\varPsi_{A}^{t}(W_{AA}-\varSigma_{WAA})\varPsi_{A}\|_{2}.
\end{align*}
\normalsize
Since $\|A\|_{2}\le \|A\|_{\infty,2}$ and $\|AB\|_{\infty,2}\le\|A\|_{\infty}\|B\|_{\infty,2}$, it follows that on $\mathcal{E}$
\begin{align*}
\|\Upsilon\|_{2}&\le \|\hat V_{A}-\varPsi_{A}\|^{2}_{\infty,2}\|W_{AA}-\varSigma_{WAA}\|_{\infty}+2\|\varPsi_{A}\|_{\infty,2}\|\hat V_{A}-\varPsi_{A}\|_{\infty,2}\|W_{AA}-\varSigma_{WAA}\|_{\infty}\\
&+\|\hat V_{A}-\varPsi_{A}\|^{2}_{\infty,2}\|\varSigma_{WAA}\|_{\infty}+2\|\varPsi_{A}\|_{\infty,2}\|\varSigma_{WAA}\|_{\infty}\|\hat V_{A}-\varPsi_{A}\|_{\infty,2}+\|\varPsi_{A}\|_{\infty,2}^{2}\|W_{AA}-\varSigma_{WAA}\|_{\infty}\\
&\le \|W_{AA}-\varSigma_{WAA}\|_{\infty}\left( \|\hat V_{A}-\varPsi_{A}\|^{2}_{\infty,2}+2\|\varPsi_{A}\|_{\infty,2}\|\hat V_{A}-\varPsi_{A}\|_{\infty,2}+\|\varPsi_{A}\|_{\infty,2}^{2}\right)\\
&+2\|\varPsi_{A}\|_{\infty,2}\|\varSigma_{WAA}\|_{\infty}\|\hat V_{A}-\varPsi_{A}\|_{\infty,2}.
\end{align*}
Therefore, there exists constant $C''$ such that on the event $\mathcal{E}$
\begin{equation*}
\|\Upsilon\|_{2}\le C''\sqrt{\frac{\log(ps)s^{2}}{N}}.
\end{equation*}
If $\|\Upsilon\|_{2}\|\varLambda^{-1}\|_{2}<1$, then
\begin{equation*}
\|\hat \varLambda^{-1}-\varLambda^{-1}\|_{2}\le \|\varLambda^{-1}\|_{2}^{2}\|\Upsilon\|_{2}^{2}(1-\|\Upsilon\|_{2}\|\varLambda^{-1}\|_{2})^{-1}.
\end{equation*}
It follows that on the event $\mathcal{E}$, for each $g\in \{1,...,G\}$ there exists constant $C$ such that
\begin{equation*}
|h^{g}-\hat h^{g}|\le  C\sqrt{\frac{\log(ps)s^{2}}{N}}.
\end{equation*}
Under (C1) and (C2), $P(\mathcal{E}_{1})\to 1$ by Hoeffding inequality, $P(\mathcal{E}_{2})\to 1$ from the first part of Corollary 2, $P(\mathcal{E}_{3})\to 1$ and $P(\mathcal{E}_{4})\to 1$ from the proofs of Lemmas 2 and 3, $P(\mathcal{E}_{5})\to 1$ from Corollary~1. Therefore, $P(\mathcal{E})\to 1$, hence
\begin{equation*}
\hat h^{g}\overset{P}{\longrightarrow}h^{g}.
\end{equation*}
By the Continuous Mapping Theorem \citep[Theorem~2.3]{vanderVaart:2000td}, this implies that for any $g_{1},g_{2}$
\begin{equation*}
\hat h^{g_1}-\hat h^{g_2}\overset{P}{\longrightarrow}h^{g_1}-h^{g_2}.
\end{equation*}
Without loss of generality, assume $h_{\varPsi}(x)=1$. In other words, $h^{1}< h^{g}$ for all $g\in\{2,...,G\}$. Then
\begin{align*}
P\left(\hat h_{\hat V}(x)=h_{\varPsi}(x) \right)&=P\left(\hat h^{1}\le \hat h^{g}\mbox{ for all }g\neq 1\right)\\
&=P\left(h_{1}-h_{g}+\left(\hat h^{1}-\hat h^{g}-(h^{1}-h^{g})\right)\le0\mbox{ for all }g\neq 1\right)\\
&=P\left(\left(\hat h^{1}-\hat h^{g}-(h^{1}-h^{g})\right)\le h^{g}-h^{1}\mbox{ for all }g\neq 1\right).
\end{align*}
Since $\hat h^{1}-\hat h^{g}\overset{P}{\longrightarrow}h^{1}-h^{g}$, it follows that $P\left(\hat h_{\hat V}(x)=h_{\varPsi}(x) \right)\to 1$.
\end{proof}

{\singlespacing
\bibliography{FLDA}
\bibliographystyle{chicago}}
\end{document}

%% file: Notation.tex
\usepackage{amsmath}
\usepackage{textcomp} 
\usepackage{natbib}
\usepackage{amsfonts}
\usepackage{amssymb}
\usepackage{amsthm}
\usepackage{graphicx}
\usepackage[colorlinks,citecolor=blue,urlcolor=blue]{hyperref}
\usepackage{cite}
\usepackage[mathscr]{eucal}
\usepackage{fullpage}
\usepackage{bbm}
\usepackage{latexsym}
\usepackage{verbatim} 
\usepackage{calc}
\usepackage{perpage}
\usepackage{marvosym}
\usepackage{hieroglf}
\usepackage{multirow} 
\usepackage{rotating} 
\usepackage{dsfont}

\usepackage{color}
\usepackage{tikz}

\newtheorem{thm}{Theorem}
\newtheorem{Claim}{Proposition}
\newtheorem{lemma}{Lemma}
\newtheorem{corollary}{Corollary}

\DeclareMathOperator{\Tr}{Tr}
\DeclareMathOperator*{\vect}{vec}

\DeclareMathOperator*{\cov}{Cov}
\DeclareMathOperator*{\card}{card}

\def\Rho{P}

\def\E{\mathbb{E}}

\def\R{\mathbb{R}}

